\documentclass[a4paper,11pt]{article}

\usepackage{jheppubmod} 
\usepackage{comment}
\usepackage{cite}
\usepackage{color}
\usepackage{psfrag}
\usepackage{array}
\usepackage{amssymb}
\usepackage{amsmath}
\usepackage{amsthm}
\usepackage{mathtools}
\usepackage[T1]{fontenc} 
\usepackage{cancel}
\usepackage{cleveref}
\crefformat{section}{\S#2#1#3} 
\crefformat{subsection}{\S#2#1#3}
\crefformat{subsubsection}{\S#2#1#3}
\usepackage[most]{tcolorbox}

\usepackage[labelsep=quad]{subcaption}

\usepackage{epsfig}
\usepackage{tikz}
\usetikzlibrary{arrows,plotmarks,positioning,calc}
\usepackage{graphicx}
\usepackage{physics}
\usepackage{float}

\theoremstyle{definition}

\newtheorem{theorem}{Theorem}[section]
\newtheorem{lemma}[theorem]{Lemma}

\theoremstyle{remark}

\title{Positive Geometries of S-matrix without Color} 

\author[a]{Mrunmay Jagadale}\emailAdd{mjagadal@caltech.edu}
\author[b]{Alok Laddha,}\emailAdd{aladdha@cmi.ac.in}

\affiliation[a]{Walter Burke Institute for Theoretical Physics,  California Institute of Technology, Pasadena, CA 91125, USA}
\affiliation[b]{Chennai Mathematical Institute, H1, SIPCOT IT Park, Siruseri, Kelambakkam 603103, India}

\abstract{
In this note, we prove that the realization of associahedron discovered by Arkani-Hamed, Bai, He, and Yun (ABHY) is a positive geometry for tree-level S-matrix of ordinary $\phi^{3}$ theory without color. More in detail, we consider diffeomorphic images of the ABHY associahedron. The diffeomorphisms are linear maps parametrized by the right cosets of the Dihedral group $D_{n}$. The set of all the boundaries associated with these copies of ABHY associahedron exhaust all the poles of $\phi^{3}$ theory. We prove that the sum over the diffeomorphic copies of ABHY associahedron is a positive geometry and the total volume obtained by summing over all the dual associahedra is proportional to the tree-level S matrix of (massive or massless) $\phi^{3}$ theory.  We then provide non-trivial evidence that the projection of the $\frac{n-4}{2}$ $d$-log forms (parametrized by the accordiohedron known as Stokes polytope) on these realizations of the associahedron lead to the tree-level amplitudes in $\phi^{4}$ theory without color. 

Our results build on ideas laid out in \cite{Jagadale:2021iab, Jagadale:2022rbl}, leading to further evidence that a large class of positive geometries which are diffeomorphic to the ABHY associahedron defines an ``amplituhedron" for  tree-level S matrix of \emph{some} local and unitary scalar theory.   An interesting offshoot of our analysis is the CHY formula for the tree-level amplitude in $\phi^{3}$ theory without color.  We also highlight a fundamental obstruction in applying these ideas to discover positive geometry for the un-colored  $\phi^{3}$ S-matrix integrand at one-loop.} 

\begin{document} 
\maketitle
\flushbottom

\section{Introduction}
The ``amplituhedron program" of the S-matrix \cite{Arkani-Hamed:2013jha}, \cite{Ferro_2021},\cite{Herrmann:2022nkh}(and references therein) has offered a number of remarkable insights in deepening our understanding of the analytic structure of the scattering amplitudes by geometrizing the analytic structure of the S-matrix. In the landscape of scalar field theories, these insights have led to the discovery of a fundamental postulate known as projective invariance, from which unitarity and locality emerge as a set of derived postulates.   

By recasting S-matrix as a differential form in the kinematic space ${\cal K}_{n}$, the amplituhedron program has consolidated our understanding of the  recursion relations first discovered in the context of MHV amplitudes in gauge theories by Andrew Hodges \cite{Hodges:2013aa} which was further developed by Arkani-Hamed, Bourjailly, Cachazo Hodges, and Trnka \cite{Arkani-Hamed:2012aa}.  It has enhanced our understanding of the worldsheet formulation of the S matrix encapsulated in the CHY (Cachazo, He, and Yuan)  formula by identifying  certain compactifications of the world-sheet moduli space with a specific Positive geometry known as associahedron in the kinematic space and finally, it has revealed a potentially deep connection between geometrization of the S-matrix in the kinematic space and the color-kinematics duality. The fundamental thesis behind these developments is the discovery of a class of polytopes known as  positive geometry whose boundaries capture \emph{all} the poles of the planar S matrix of scalar quantum field theories. The canonical forms associated with positive geometries are amplitudes in a QFT. 

Arkani-Hamed, Bai, He, and Yan discovered the first example of positive geometry (known as associahedron) directly  inside the kinematic space, \cite{Arkani-Hamed:2017mur}. The unique canonical form defined by the associahedron in ${\cal K}_{n}$ is the amplitude of bi-adjoint $\phi^{3}$ theory.  However, we now know that the associahedron polytope is the first member of an infinite family of positive geometries known as accordiohedra. Soon after the discovery of associahedron in the kinematic space, specific realizations of accordiohedra were also discovered, which defined  color-ordered S-matrix for scalars without derivative coupling, \cite{He:2018okq, Banerjee:2018tun, Raman:2019utu, Herderschee:2020lgb, Herderschee:2019wtl, Salvatori:2019phs, Kalyanapuram:2019nnf, Aneesh:2019cvt, Damgaard:2021ab, Barmeier:2021iyq, Cachazo:2022voc}.

One of the primary reasons why the discovery of positive geometries in kinematic space has led to a deeper understanding arises from parallel developments in the subject of quiver representation theory and cluster algebra. Namely, given  the set of all possible dissection of an $n$-gon (which may include puncture in the interior) with a fixed dimension, it generates, on the one hand, combinatorial geometries such as the accordiohedron and on the other hand, a family of vectors in an abstract Cartesian space such that these vectors form a simplicial and a complete fan, \cite{Padrol2019AssociahedraFF}. 
The polytopal realization of this fan is then a convex realization of the combinatorial polytope. 

In  the seminal paper \cite{Bazier-Matte:2018rat}, it was shown that in the case of associahedron, such a polytopal realization of the aforementioned fan  matches precisely the ABHY associahedron. This result was extended to include realizations arising from so-called cyclic quivers and generate convex realizations of accordiohedron in \cite{Padrol2019AssociahedraFF}.

Thus from combinatorics of abstract dissections and a single non-trivial principle known as projective invariance, one obtains S-matrix for an infinite family of scalar quantum field theories, clearly hinting at a possibility that the S-matrix of a local QFT may simply be volumes of certain polytopes where the measure is determined by a deeper principle such as projective invariance. 

As we argued in \cite{Jagadale:2021iab, Jagadale:2022rbl}, the geometric data in ABHY associahedron is even richer than what the seminal developments revealed. The new structures emerged from a simple observation that the Cartesian space in which quiver algebra shaped the polytopes such as accordiohedron and the physical kinematic space of Mandelstam invariants are a priori distinct vector spaces. The simplest identification between them generated color-ordered tree-level amplitudes of massless scalar QFTs.\footnote{Strictly speaking, this result is valid as long as the dimension of spacetime $D\, \geq\, n$ where $n$ is the number of external particles, as in $D\, =\, 4$ the kinematic space is a $3n-10$ dimensional variety in ${\cal K}_{n}$ on which the so-called Gram conditions are satisfied. \cite{Damgaard:2021aa},\cite{ukowski:2022aa}.} However, this identification is just one such map among the space of all diffeomorphisms between two copies of $P_{n}({\bf R})$. A rather simple class of diffeomorphisms, namely linear maps between $P^{n}({\bf R})$ to $P^{n}({\bf R})$ led to a deformation of  the ABHY associahedron in the kinematic space which turned out to be positive geometries for  theories far removed from massless scalar field theories.  This rather simple observation potentially opens up avenues to answer two outstanding questions in the program. 
\begin{enumerate}
\item What is the positive geometry for a manifestly crossing symmetric (as opposed to planar) tree-level S-matrix, such as the ordinary massive $\phi^{3}$ scalar field theory? 
\item Can the scattering forms associated with amplitudes with non-trivial numerator factors be derived from fundamental postulates like projectivity? 
\end{enumerate}
We answer the first question for $\phi^{3}$ and $\phi^{4}$ theories without color in this note. 

More in detail, we show that there exists a $\frac{(n-1)!}{2}$ linear maps $f_{[\sigma]}$ that are parametrized by the elements $[\sigma]$ of right cosets of the dihedral group ${\cal D}_{n}$.  $f_{[\sigma]}$  maps the ABHY associahedron in the Cartesian space to an associahedron inside the kinematic space. We will refer to the images as deformed realizations of the associahedron and prove that the sum of the canonical forms over all of these associahedra generates the complete (as opposed to color-ordered) S matrix of a $\phi^{3}$ theory with mass $m$. 

The organization of the paper is as follows. In section (\ref{bose}), we use the permutation symmetry to sub-divide the family of all the poles of a $\phi^{3}$ S-matrix in ``associahedron slices". Such a classification helps us in deriving a beautiful formula that relates the total number of Feynman diagrams with the cardinality of the vertex set of the associahedron, which is given by a Catalan number.  In section (\ref{review}), we briefly review the construction of ABHY associahedron and so-called generalized permutahedron, which have hitherto been considered as positive geometries whose vertices are associated with non-planar poles.  In section (\ref{pimek}), we define the linear maps between the embedding space and the kinematic space, which are parametrized by elements of the Bose symmetry and obtain deformed realizations of the ABHY associahedron. In section \ref{dfiks}, we pull back the canonical form from the cartesian space ${\cal E}_{n}$  to ${\cal K}_{n}$ and show how the sum overall such forms are proportional to crossing symmetric S-matrix of $\phi^{3}$ theory. 

In section \ref{chynp}, we use the linear maps between ${\cal E}_{n}$ and ${\cal K}_{n}$ to write the CHY formula for crossing symmetric $\phi^{3}$ S-matrix where the CHY integrand is the Parke-Taylor form. In section \ref{phi4}, we argue that a weighted sum over the projection of $\frac{n-4}{2}$ ranked $d\ln$ forms, which are labeled by quadrangulations of the $n$-gon is the tree-level scalar amplitude with quartic interactions. Once again, the amplitude is manifestly crossing symmetric as the scalar particles have no color. In section \ref{1loop}, we show how these rather simple ideas face a fundamental obstruction in relating the $\hat{D}_{n}$ polytope discovered in \cite{afpst} and one loop integrand of $\phi^{3}$ S matrix without color. We end with a discussion of our results and place them in the context of the entire program of geometrizing the S-matrix. 

\section{Combinatorics of poles of scattering amplitude in scalar theory}\label{bose}
The poles of $n$-point tree-level scattering amplitude in a local scalar field theory are of the form $\frac{1}{s_{I}}$, where $s_{I} =\left( \sum_{i \in I} p_{i}  \right)^{2}$, and $1< |I|<n-1$. There are $2^{n-1}- n-1$ such poles. However, not all poles can come in a single scattering channel of a unitary theory. We say two poles are compatible if they can come together in a single scattering channel of a unitary theory. For example, at 5-point, the poles $1/s_{12}$ and $1/s_{13}$ can never come together in a single scattering channel, but the poles $1/s_{13}$ and $1/s_{123}$ can come in a single scattering channel. This notion of compatibility of poles gives a combinatorial structure to the set of poles of $n$-point tree-level scattering amplitude in a local unitary scalar field theory. We will consider scattering amplitudes in $\phi^{3}$-theory where all poles that can appear do appear. 

When restricted to planar poles, the combinatorial structure mentioned above is the combinatorial associahedron. Starting with the combinatorial structure of an associahedron, we can write down algebraic relations between planar kinematic variables, viz. $X_{i,j} + X_{i+1, j+1}- X_{i,j+1}-X_{i+1,j}=c_{ij}$, which in turn gives a geometric realization of the associahedron in the kinematic space, viz. the ABHY associahedron. This ABHY associahedron gives us the scattering amplitude of planar $\phi^{3}$-theory. However, as we will see in this section, the underlying combinatorial structure is more universal than previously thought. This, in turn, tells us that the underlying algebraic and geometric structures are more universal and can capture scattering amplitudes of $\phi^{3}$-theory (without color), which we will see in the latter sections. 

The combinatorial structure on the set of all poles is such that this set can be neatly divided  into a family of associahedra intertwined. For example, at 4-point, the set of all poles is $\{1/s_{12},1/s_{23}, 1/s_{13}\}$. This set is made of three one-dimensional associahedra, viz. $\{ s_{12} , s_{23}\}$, $\{ s_{23} , s_{13}\}$, and $\{ s_{13} , s_{12}\}$. Here, each associahedron shares its poles with the other two. For higher-point amplitudes, the associahedra are highly intertwined. This poses the question, how to get all these different associahedra slices of the set of all poles? To understand this, we have to look at the action of permutation(Bose) symmetry, $S_{n}$, on the set of all poles of $n$-point scattering amplitudes. 

There is a natural action of Bose symmetry on the set of all poles. This action induces the action of Bose symmetry on scattering channels or Feynman diagrams. The notion of compatibility discussed above commutes with the action of Bose symmetry. Two poles of scattering amplitude are compatible if and only if their image under the action of any element $\sigma \in S_{n}$ are compatible. Therefore, the action of any element $\sigma \in S_{n}$ on the combinatorial associahedra formed by planar poles takes us to another associahedra slice of the set of all poles. Given any Feynman diagram, there exists a $\sigma \in S_{n}$ such that the action $\sigma$ on that Feynman diagram gives us a planar Feynman diagram. Therefore, by acting different elements of $S_{n}$ on the combinatorial associahedron of planar poles, we get all the ``associahedron slices'' of the set of all poles. 
\begin{figure}
    \centering
    \includegraphics[scale=0.37]{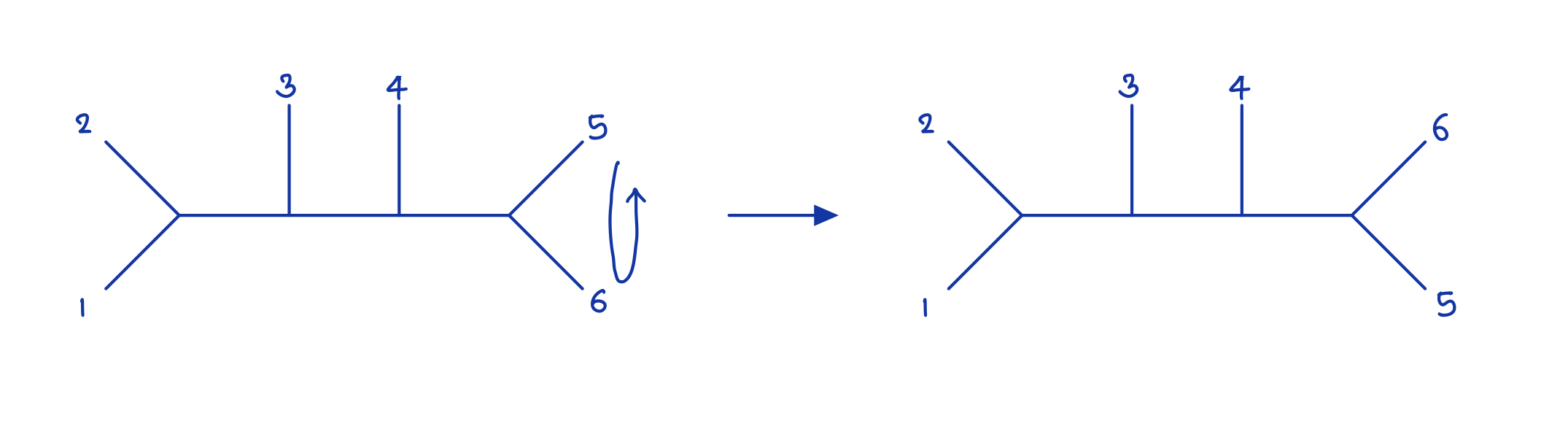}
    \caption{Bose symmetry that leaves the scattering channel invariant.}
    \label{Bosesymmstab}
\end{figure}

Generically the action of an element of $S_{n}$ on a planar scattering channel takes us to a non-planar scattering channel. However, there is a subgroup of $2^{n-2}$ elements generated by elements like the one depicted in figure \ref{Bosesymmstab} that give you back the same scattering channel.
More in detail, suppose $\sigma_{S}$ is an element of $S_{n}$ that takes a planar diagram $P$ to $S$ when acted by $\sigma_{S}$. That is, $\sigma_{S} \cdot P = S $. Then $\alpha \cdot \sigma_{S} \cdot P = S$ as well, where $\alpha$ is one of the $2^{n-2}$ elements that takes $S$ back to itself. At the same time, there are $n$ order $n$ elements of $S_{n}$, the cyclic permutations, that take you from one planar scattering channel to another planar scattering channel. Therefore, if $\sigma_{P,P^{\prime}} \in C_{n}$ is such that $\sigma_{P,P^{\prime}} \cdot P^{\prime} = P$, where $P^{\prime}$ is planar, then $\alpha \cdot \sigma \cdot \sigma_{P,P^{\prime}} $ takes $P^{\prime}$ to $S$, $ \alpha \cdot \sigma \cdot \sigma_{P,P^{\prime}} \cdot P^{\prime} =S$. In other words, for any scattering channel $S$, there are $ n 2^{n-2}$ elements of $S_{n}$ that take some planar diagram to $S$. This is beautifully captured by the formula relating the total number of distinguished Feynman diagrams with cubic vertices and the Catalan number, which is the number of planar Feynman diagrams.
\begin{equation}
    (2n-5)!! = \frac{C_{n-2} (n-1)! }{2^{n-2}},
\end{equation}
where $(2n-5)!!$ is the total number of $n$-point tri-valent Feynman diagrams and $C_{n-2}$ is the total number of $n$-point tri-valent planar Feynman diagrams.\footnote{This identity follows trivially from the definitions, $C_{n-2}\, =\, \frac{(2n-4)!}{(n-2)!(n-1)!}$, $(2n - 5)!!\, =\, \frac{(2n-4)!}{2^{n-2}(n-2)!}$.}

\section{ABHY associahedron and Zonotopal Generalised Permutahedron}\label{review}
Consider the Cartersian space ${\bf R}^{\frac{n(n-3)}{2}}$. We denote it as ${\cal E}_{n}$, so as to contrast it with the kinematic space of Mandelstam invariants ${\cal K}_{n}$.\footnote{To be completely precise, ${\cal E}_{n},\, {\cal K}_{n}$ are real projective spaces of dimension $\frac{n(n-3)}{2}$, but we gauge fix the projective freedom by choosing the first co-ordinate to be $1$.} We fix once and for all, a basis $\{x_{ij}\, \vert\,  1\, \leq\, i\, <\, j + 1\, \leq\, n\, \}$ in  ${\cal E}_{n}$. This basis is  labelled by chords of an n-gon :  $\{\, x_{ij}\, \vert\, 1\leq\, i\, < j + 1\, \leq\, n \}$.  Let $y_{i_{1}, \dots\, i_{k}}\, 1\, \leq\, i_{1}\, <\, \dots\, <\, i_{k}\, \leq\, n$ be defined as, 
\begin{align}\label{yijxij}
    y_{ij}\, &:=\, x_{i,j+1}\, +\, x_{i+1,j}\, -\, x_{ij}\, -\, x_{i+1,j+1} \endline
    y_{i_{1}i_{2}\, \dots\, i_{k}}\, &:=\, \sum_{m,n=1\vert\, m < n}^{k}\, y_{i_{m}i_{n}}
\end{align}
These two equations imply that,
\begin{flalign}
x_{ij}\, =\, y_{i,i+1,\dots,j-1}
\end{flalign}
The $(n-3)$-dimensional ABHY associahedron \cite{Arkani-Hamed:2017mur} is a specific convex realization of the associahedron $A_{n}$ in the positive quadrant of the embedding space, ${\cal E}_{n}^{\geq\, 0} :=\, \{\, x_{ij}\, \geq\, 0\, \forall\, (i,j)\, \}$, defined as follows \cite{Arkani-Hamed:2017mur, Bazier-Matte:2018rat,Padrol2019AssociahedraFF}; Given any reference triangulation $T$,  consider the intersection of the positive quadrant of the embedding space $\mathcal{E}_{n}^{\geq 0}$ and the hyper-planes given the equations,
\begin{flalign}\label{yijcij}
y_{ij} = - c_{ij}\, \hspace{0.1cm} \vert\, \hspace{0.1cm} c_{ij}\, >\, 0\, \hspace{0.1cm} \forall\, \hspace{0.1cm} (i,j)\, \notin\, T^{c}.
\end{flalign}
For any choice of positive constants $c_{ij}$, we get a polytopal realization of $A_{n}$. 

We will denote the ABHY realization by $A_{n}^{T}$ as it explicitly depends on the choice of the reference triangulation. Our review of ABHY realization is based on \cite{Padrol2019AssociahedraFF}, whereas  the original construction in \cite{Arkani-Hamed:2017mur} where these polytopes were discovered directly in the kinematic space of Mandelstam invariants ${\cal K}_{n}$. We now quickly review the relationship between the two. 

In space-time dimensions $d\, \geq\, n-3$, the kinematic space  ${\cal K}_{n}$ spanned by the Mandelstam variables can also be realized as a ${\bf R}^{\frac{n(n-3)}{2}}$ space.  This identification involves identifying a (complete) set of Mandelstam invariants with the coordinate basis. A rather convenient choice of co-ordinates is the set of planar kinematic variables $\{\, X_{ij}\, \vert\, 1\, \leq\, i\, <\, j + 1\, \leq\, n\, \}$. 
\begin{flalign}
X_{ij}\, =\, (p_{i}\, +\, \dots\, +\, p_{j-1})^{2}
\end{flalign}
In general, given any ordering $\{\, \sigma(1),\, \dots,\, \sigma(n)\, \}$ obtained by action of a permutation $\sigma\, \in\, S_{n}$ on the standard ordering, one can define the so-called $\sigma$-planar variables which are labelled by chords of an $n$-gon whose vertices are ordered in clockwise direction as $(\, \sigma(1),\, \dots,\, \sigma(n)\, )$.  Any such choice of $\sigma$-planar basis defines for us a linear space ${\cal K}_{n}^{\sigma}$ with $\sigma$-planar variables defining an orthonormal basis.\footnote{We note that in ${\cal K}_{n}$,  $\sigma$-planar variables do not form the Cartesian basis if $\sigma$ is not a $\{\, i,\, i+1,\, \dots,\, n,\, 1,\, \dots,\, i-1\}$.}
In what follows, we fix once and for all an ordering of $( 1, \dots,n )$, and ${\cal K}_{n}$ is defined as the kinematic space in which $X_{ij}$ form an orthonormal coordinate system. 

The relationship between ABHY realisation in \cite{Arkani-Hamed:2017mur} and \cite{Padrol2019AssociahedraFF} is simply the identifcation of ${\cal E}_{n}$ with ${\cal K}_{n}$ by the isometry 
\begin{flalign}\nonumber
x_{ij}\, =\, X_{ij}\, \hspace{0.1cm} \forall \hspace{0.1cm} (i,j).
\end{flalign}
This realization has several rather remarkable properties that distinguish  it from other realizations of the associahedron, making  it the ``amplituhedron" for bi-adjoint $\phi^{3}$ theory \cite{Arkani-Hamed:2017mur}. Let ${\cal D}$ be the set of all the dissections of an $n$-gon whose vertices are labeled $\{1,\, \dots,\, n\}$ in a clockwise direction. 
\begin{itemize}
\item All the co-dimension one faces (facets) of $A_{n-3}^{T}$ are in bijection with the set 
\begin{flalign}\nonumber
\{\, X_{ij}\, =\, 0 \vert (i,j)\, \in\, {\cal D}\, \}
\end{flalign}
which is the set of all the simple poles in the bi-adjoint color-ordered amplitude, ${\cal M}_{n}(\, (1, \dots, n)\, \vert\, (1, \dots, n)\, )$. 
\item For a $n$ dimensional associahedron, precisely $[\frac{n}{2}]$ pairs of co-dimension one faces are parallel to each other. 
\item On a facet of $A_{n-3}^{T}$ which corresponds to $X_{ij}\, =\, 0$ for some $(i,j)\, \in\, {\cal D}$, $X_{mn}\, > 0 $ $\forall\, (m,n)\, \in\, {\cal D}$ such that $(m,n)\, \cap\, (i,j)\, \neq\, 0$. 
\end{itemize}
For later purposes, we also recall that  $A_{n-3}$ also induces a unique \emph{projective} $(n-3)$-form on ${\cal E}_{n}$.
\begin{flalign}
\Omega^{{\cal E}_{n}}_{n-3}\, =\, \sum_{v\, \in\, A_{n-3}}\, (-1)^{\sigma_{v}}\, \bigwedge_{(ij)\, \in\, v}\, \dd \log x_{ij}
\end{flalign} 
The projectivity is the invariance of $\Omega^{{\cal E}_{n}}_{n-3}$ under $x_{ij}\, \rightarrow\, f(\{x_{mn}\})\, x_{ij}$ for any function $f$.  As ABHY proved, under the identification of ${\cal E}_{n}$ with ${\cal K}_{n}$ $A_{n-3}^{T}$ is the positive geometry for bi-adjoint tree-level S matrix as 
\begin{flalign}
\Omega^{K_{n}}_{n-3}\vert_{A_{n-3}^{T}}\, =\, {\cal M}_{n}(\, (1,\, \dots\, n)\, \vert\, (1,\, \dots,\, n)\, )\, \bigwedge_{(ij)\, \in\, T}\, \dd, X_{ij}
\end{flalign}
In this paper, we use these ideas to look for positive geometries in ${\cal K}_{n}$ which generate the S-matrix of (massive or massless) $\phi^{3}$ theory without color in $D\, \geq\, n$ dimensions.  That is, our goal is to discover a set of linear maps from ${\cal E}_{n}\, \rightarrow\, {\cal K}_{n}$ such that the resulting  image of  the ABHY associahedron in ${\cal K}_{n}$ constitute a family of polytopes whose co-dimension one facets exhaust all the poles of the tree-level S matrix with cubic interaction.  

We now review the discovery of a family of polytopes known as Permutahedra in the kinematic space whose canonical form has poles at $\{\, s_{1\sigma(2)\, \dots\, \sigma(i)}\, \vert\, 2\, \leq\, i\, \leq\, n-1\, \vert\, \sigma\, \in\, S_{n-2}\, \}$. 

Starting from the seminal work in  \cite{Arkani-Hamed:2017mur}, it has been realized that a simple polytope whose vertex set is bijection with the set of poles defined above is the permutahedron ${\cal P}_{n}$ \cite{Arkani-Hamed:2017mur,Early:2018zuw}. This is because, as we review below,  ${\cal P}_{n}$ is a combinatorial polytope whose vertex set is in bijection with $S_{n-2}$, 
\cite{2005math......7163P}

 Permutahedron can be defined as follows; \footnote{In the interest of brevity, we are paraphrasing the definition such that it relates to the primary question under investigation.} Given a set $I\, =\, \{1\, \dots,\, n\}$, we fix two elements, say $i,\, i + 1$ as the ``initial" and ``final" vertices and ``rotate $I$" to
 \begin{flalign}
 I\, =\, \{\, i+1,\, i+2,\, \dots,\, n,\, 1,\, 2,\, \dots,\, i\, \}
 \end{flalign}
 Let 
 \begin{flalign}
 I_{i}\, =\, I\, -\, \{i, i+1\}
 \end{flalign}
 We now consider the configuration space defined as
 \begin{flalign}
 \sigma\, \cdot\, I\, =\, \{\, i+1,\, \sigma\, \cdot\, I_{i},\, i\, \}\, \hspace{0.1cm} \forall\, \hspace{0.1cm} \sigma\, \in\, S_{n-2}
 \end{flalign}
We now define a permutahedron ${\cal P}_{n}^{i+1,i}$ as a simple polytope whose vertex set is in bijection with $\sigma\, \cdot I$ and two vertices are adjacent if and only if they are related by a $\sigma$ which exchanges $(m,n)\, \rightarrow\, (n,m)\, \vert\, (m,n)\, \in\, I_{i}$.  ${\cal P}_{n}$ has several interesting properties.
\begin{enumerate}
\item It is a simple polytope with the dimensionality given by the number of propagators in a tri-valent Feynman graph, $n-3$
\item It is a member of the family of polytopes known as Cayley polytopes with the largest number ($(n-2)!$) of vertices among this family. 
\item Given a planar ordering, Each permutahedron has precisely one vertex that corresponds to the planar channel, and all the other vertices correspond to non-planar Feynman diagrams.
\item No vertex of ${\cal P}_{n}^{i+1,i}$ corresponds to the set of propagators of the type 
\begin{flalign}\nonumber
\{\, (i+1,\sigma(i+2)),\, ( \sigma(i+3),\, \sigma(i+4) ),\, \dots,\, (\sigma(i-1), i)\, \}
\end{flalign}
\end{enumerate}
In \cite{Arkani-Hamed:2017mur}, Arkani-Hamed, Bai, He, and Yan, in fact, constructed a convex realization of the Permutahedron in ${\cal K}_{n}$ as follows. 

Given ${\cal P}_{n}^{i+1,i}$, consider a positive orthant  in ${\cal K}_{n}$ defined using the following inequalities.
\begin{flalign}
s_{i+1\sigma(i+2)\, \dots\, \sigma(i+k)}\, \geq\, 0\, \forall\, k \leq\, n-2, \sigma\, \in\, S_{n-2}
\end{flalign}
The set of constraints whose intersection with  this positive orthant maps out the generalized permutahedron are simply, 
\begin{flalign}
s_{ij}\, =\, -\, c_{ij}\, \forall\, i+2\, \leq\, i\, <\, j\, \leq\, i-1
\end{flalign}
The canonical form on the convex ${\cal P}_{n}^{i+1,i}$ is, 
\begin{flalign}
\Omega_{n-3}\vert_{{\cal P}_{n}^{i+1,i}}\, =\, \sum_{\sigma\, \in\, S_{n-2}}\, [\, \prod\, \frac{1}{s_{i+1\sigma(i+2)\, \dots\, \sigma(i-1)}}\, ]\, \bigwedge_{j=i+2}^{i-1}\, \dd s_{i+1\, j}
\end{flalign}
which is a partial contribution to the S-matrix of $\phi^{3}$ theory.  In a beautiful paper  \cite{Early:2018zuw}, Nick Early analyzed these realizations in further detail and proved that they are, in fact, equivalent to a class of simple polytopes called zonotopal generalized permutahedra which were discovered by A. Postnikov \cite{2005math......7163P}.\footnote{This relationship was already anticipated by the authors in \cite{Arkani-Hamed:2017mur}}. 

In spite of the remarkable richness contained in these geometries (see, e.g.\cite{He:2018okq}) as well as the fact that each such ${\cal P}_{n}^{i+1,i}$ has precisely one planar and the rest non-planar channels as vertices, a precise relationship between zonotopal generalized permutahedra  and the S-matrix of ordinary $\phi^{3}$ theory remains to be understood. Property (4) stated above only adds a further layer of mystery in the role permutahedra may play in unraveling the structure of S-matrix without color. 

We will not consider these mythical objects in our work and will instead show how to recover all the poles (planar as well as non-planar) of the S matrix from a set of incarnations of the ABHY associahedron located in various quadrants of ${\cal K}_{n}$. 
\
\section{Permutation induced maps between ${\cal E}_{n}$ and ${\cal K}_{n}$}\label{pimek}
In this section, we define a family of maps between  the embedding space ${\cal E}_{n}$ and the kinematic $\mathcal{K}_{n}$ space parameterized by the Bose symmetry reviewed in section \ref{bose}. Given an element $\sigma \in S_{n}$ we consider the following linear isomorphism, \footnote{In \cite{Arkani-Hamed:2017mur} $s_{\sigma(i)\sigma(i+1)\ldots\sigma(j-1)}$ were called the $\sigma$-planar variables $X_{\sigma(i)\sigma(j)}$.}
\begin{align}\label{yijsig}
    f_{\sigma} : \mathcal{E}_{n} & \rightarrow \mathcal{K}_{n} \endline 
                        x_{ij}   & \mapsto s_{\sigma(i)\sigma(i+1)\ldots\sigma(j-1)}.
\end{align}
As we have fixed the basis in both, ${\cal K}_{n}$ and ${\cal E}_{n}$, each such map $f_{\sigma}$ is a matrix in $GL(\frac{n(n-3)}{2})$. 

The linear isomorphism $f_{\sigma}$ defined in \eqref{yijsig} maps the ABHY associahedra $A_{n}^{T}$ in the (positive quadrant of) the embedding space to a geometric realization of associahedra in kinematic space $\mathcal{K}_{n}$. Given a triangulation $T$ of an $n$-gon and an element $\sigma \in S_{n}$, we get such a geometric realization of associahedra. We denote this geometric realization as $A_{n}^{T,\sigma}$.

More in detail, as discussed in section \ref{review}, given a triangulation, we have a geometric realization of associahedra by considering the intersection of hyper-planes given in equation \eqref{yijcij} with the positive quadrant of the embedding space. Now the linear isomorphism $f_{\sigma}$ maps the hyper-planes to the hyper-planes 
\begin{equation}\label{sigmasijcij}
    s_{\sigma(i)\sigma(j)} = -c_{ij} \hspace{0.3cm} \forall \hspace{0.1cm} (i,j) \hspace{0.1cm} \notin \hspace{0.1cm} T^{c}.
\end{equation}
While the positive quadrant of the embedding space $\mathcal{E}_{n}^{\geq 0}$ is mapped to what we call $\sigma$-positive region $\mathcal{K}_{n}^{\sigma \geq 0} \subset \mathcal{K}_{n}$. The $\sigma$-positive region is given by 
\begin{equation}\label{sigmapositive}
    \mathcal{K}_{n}^{\sigma \geq 0} := \{ s_{\sigma(i),\, \sigma(i+1),\dots,\, \sigma(j)}\, \geq\, 0\, \forall\, 1\, \leq\, i\, <\, j-1\, \leq\, n-1\,   \}.
\end{equation}
The intersection of hyper-planes given by \eqref{sigmasijcij} and the $\sigma$-positive region $\mathcal{K}_{n}^{\sigma \geq 0}$ gives us the geometric realization $A_{n}^{T,\sigma}$.

We now expand on several characteristics of $A_{n}^{T, \sigma}$ for generic $\sigma \in S_{n}$, which reveal the similarities and differences of this geometry with the ABHY realization. 
\begin{itemize}
 \item  The set $F_{\sigma}$ of all the co-dimension one boundaries of $A_{n}^{T, \sigma}$ are in bijection with the following set of poles,
    \begin{flalign}
    F_{\sigma}\, \overset{1-1}{=}\, \{\, s_{\sigma(i),\, \dots,\, \sigma(j-1)}\, \vert\, 1\leq\, i\, < j - 1\, \leq\, n\, \}
    \end{flalign}
    \item It is important to note that $F_{\sigma}$ is in bijection with the set of poles of the color ordered amplitude in bi-adjoint $\phi^{3}$ theory with the ordering $(\sigma(1),\, \dots,\, \sigma(n)\, \vert\, \sigma(1),\, \dots,\, \sigma(n)\, )$. However, one crucial difference with the bi-adjoint case is worth emphasizing: The kinematic space in the case of bi-adjoint color ordered amplitude is defined via the $\sigma$-planar kinematic variables, $X_{\sigma(i), \sigma(j)}$ (for a given a color order defined by a representative $\sigma$), \cite{Arkani-Hamed:2017mur}. In contrast, in  the present case where the external states have no color,  ${\cal K}_{n}$ is fixed once and for all. Our goal is to prove that the complete tree-level S matrix is given by $n-3$ forms on ${\cal K}_{n}$.
\item As the following lemma proves, the union over all the deformed realisations of $A_{n-3}$ is a positive geometry, \cite{Damgaard:2021ab}.
\end{itemize}
\begin{lemma}\label{ansigansigprime}
For any two distinct cosets $[\sigma_{1}], [\sigma_{2}]\, \in\, {\cal G}_{n}$ $\exists\, \sigma\, \in\, [\sigma_{1}],\,  \sigma^{\prime}\, \in\, [\sigma_{2}]$ such that, 
\begin{flalign}
A_{n-3}^{T, \sigma}\, \cap\, A_{n-3}^{\sigma^{\prime}}\, =\, \{ 0 \}
\end{flalign}
\end{lemma}
\begin{proof}
Without loss of generality, let us assume that  $T = \{\, (1,3), \dots, (1,n-1)\, \}$ is a reference triangulation for the ABHY associahedron. We can always choose $\sigma, \sigma^{\prime}$ such that
\begin{flalign}
\sigma(1)\, =\, \sigma^{\prime}(1)\, =\, 1
\end{flalign}
As $\sigma, \sigma^{\prime}$ are representatives of two distinct elements of ${\cal G}_{n}$, we note that there is at least one sqeucne of length k $(i_{1} <\, i_{2},\, \dots,\,  <\, i_{k})$ in $\sigma$-ordering which is mapped to $(i_{k}, i_{1},\, \dots,\, i_{k-1} )$ in $\sigma^{\prime}$-ordering. (Note that $i_{m}, i_{m+1}$ do not have to be successive entries.) We consider one such sqeuence among those which have smallest length $k_{min}\, \geq\, 2$.  Without loss of generality, let $\sigma\, =\, \textrm{id}$. Now $\sigma^{\prime}$ can be such that either $i_{k_{min}}\, < n$ or $i_{k_{min}} = n$. We consider the two cases separately.

If $i_{k_{min}}\, <\, n$ then for the chosen $T$,  $A_{n-3}^{T,\sigma^{\prime}}$ is realised in the hyper-plane located at
\begin{flalign}
\textrm{either}\ s_{i_{1}-1,i_{1}}\, =\, -c\ \textrm{or}\ s_{i_{k_{min}}i_{k_{min}}+1}\, =\, - c^{\prime}
\end{flalign}
for some negative constant $c$ ($c^{\prime}$).
As $A_{n-3}^{T,\sigma = id}$ is realised in the quadrant in which both of these variables are $\geq\, 0$ implies that the intersection between these two associahedra is empty.

If $i_{k} = n$, then the above argument does not directly apply as the associahedra constraints which locate $A_{n-3}^{T, \sigma}$ do not impose $s_{i n} = - c_{i n}$ for any $1\, \leq\, i\, \leq\, n-2$.  To be specific, let $k_{min}\, =\, 2$ and consider  
\begin{flalign}
\sigma^{\prime}\, \circ\, (1,\, \dots,\, n)\, =\, (\, 1,\, \dots,\, n-2, n, n-1)
\end{flalign}
Then in the same right coset to which $\sigma^{\prime}$ belongs $\exists\, \sigma^{\prime\prime}$ such that
\begin{flalign}
\sigma^{\prime\prime}\, \circ\, (1\, \dots\, n)\, =\, (\, 1, n-1, n, \dots,\, 3, 2)
\end{flalign}
Clearly
\begin{flalign}
A_{n-3}^{T, \sigma}\, \cup\, A_{n-3}^{T, \sigma^{\prime\prime}}\, =\, \{0\}
\end{flalign}
as the latter is located in the quadrant $s_{1,n-1}\, \geq\, 0$ whereas the former is located in the hyperplane $s_{1,n-1} = - c_{1,n-1}$.  Similar line of argument for $k_{min}\, >\, 2$ can be readily formulated. 
This completes the proof. In $n = 4$ case, a pictorial representation of the lemma can be found in figure (\ref{4ptnonplanarassociahedra}).
\end{proof}

\subsection{A family of Scattering forms on ${\cal K}_{n}$}\label{dfiks}
The associahedron defines a unique canonical $(n-3)$ form  $\Omega(A_{n})$ on ${\cal E}_{n}$.  
\begin{flalign}
\Omega_{n-3}\, = \Omega(A_{n}) =\, \sum_{v\, \in\, A_{n}}\, (-1)^{T_{v}}\, \bigwedge_{(i,j)\, \in\, T_{v}}\, \dd \log, x_{ij}
\end{flalign}
As proved in \cite{Arkani-Hamed:2017mur}, $\Omega_{n-3}$ is the unique $\dd \log$ form which is invariant under projective transformation generated by any smooth function $f\, \in\, C^{\infty}({\cal E}_{n})$, 
\begin{flalign}
x_{ij}\, \rightarrow\, f(\{x_{mn}\})\, x_{ij}
\end{flalign}
The pullback of projectively invariant $\Omega(A_{n})$ by $\{\, f_{\sigma}^{-1}\, \vert\, \sigma\, \in\, S_{n}\, \}$  generates $n!$ projective invariant $\dd \log$ forms on ${\cal K}_{n}$. 

\begin{equation}
    f^{-1 \, \star}_{\sigma} \Omega(A_{n}) = \Omega(A_{n}^{\sigma}) = \sum_{v \in A_{n}} (-1)^{T_{v}} \bigwedge_{(i,j)\in T_{v}} \dd \log s_{\sigma(i)\sigma(i+1)\ldots\sigma(j-1)}.
\end{equation}

Poles of $\Omega(A_{n}^{\sigma})$ are the $\frac{n(n-3)}{2}$ co-dimension one hyperplanes 
\begin{flalign}
s_{\sigma(i)\, \dots\, \sigma(j-1)}\, =\, 0
\end{flalign}

One of the striking and central results in \cite{Arkani-Hamed:2017mur} was the following. Pull back of $\Omega_{n-3}$ to the ABHY associahedron $A_{n}^{T}\, \in\, {\cal E}_{n}^{\geq\, 0}$ is given by the following formula. 
\begin{flalign}
\Omega_{n-3}  \vert_{A_{n}^{T}}  =  \left[  \sum_{v  \in  A_{n}}  \prod_{(i,j)  \in  v}  \frac{1}{x_{ij}}  \right]  \bigwedge_{(m,n)  \in  T}  \dd x_{mn}  =: {\cal M}_{n}\, \bigwedge_{(m,n)  \in  T}  \dd x_{mn}
\end{flalign}
Where the rational function ${\cal M}_{n}$ is defined as,
\begin{flalign}
{\cal M}_{n}\, =\, \sum_{v\, \in\, A_{n}}\, \prod_{(i,j)\, \in\, v}\, \frac{1}{x_{ij}}
\end{flalign}
It then immediately follows that pullback of $f_{\sigma}^{-1\, \star}\, \Omega_{n-3}$ on $A_{n}^{T, \sigma}$ is obtained by simply substituting $s_{\sigma(i),\, \dots,\, \sigma(j-1)}$ for $x_{ij}$ in the above formula. 
\begin{flalign}
\left( f_{\sigma}^{-1\, \star}\, \Omega_{n-3} \right)\vert_{A_{n}^{T, \sigma}}\, =\, \sum_{v\, \in\, A_{n}}\, \left[ \prod_{(i,j)\, \in\, T_{v}}\frac{1}{s_{\sigma(i)\sigma(i+1) \dots \sigma(j-1)}}\, \right] \bigwedge_{(m,n)\, \in\, T}\, ]\, \dd  s_{\sigma(m)\, \dots\, \sigma(n-1)}\\
=:\, {\cal M}_{n}(\sigma)\, \bigwedge_{(m,n)\, \in\, T}\, \dd s_{\sigma(m)\, \dots\, \sigma(n-1)},
\end{flalign}
where 
\begin{flalign}
{\cal M}_{n}(\sigma)\, : =\, \sum_{v\, \in\, A_{n-3}}\, \left[ \prod_{(i,j)\, \in\, T_{v}}\frac{1}{s_{\sigma(i)\sigma(i+1) \dots \sigma(j-1)}}\, \right].
\end{flalign}
Note that ${\cal M}_{n}(\sigma)$ can also be interpreted as the volume of the dual associahedron which is computed using the pull back form $\Omega_{n-3}^{T, \sigma}$. 

Now let's see how we can get the tree-level scattering amplitude of $\phi^{3}$-theory. To each geometric realization $A_{n}^{T,\sigma}$, we can associate a canonical form $\Omega(A_{n}^{T,\sigma})$. Naively if we sum over all such forms, we get $\sum_{T} \sum_{\sigma \in S_{n}} \Omega(A_{n}^{T,\sigma})$. However, we are grossly over-counting in this sum. The following two lemmas tell us how to get rid of these redundancies. 

\begin{tcolorbox}[colback=black!5!white,colframe=black!75!black]
\begin{lemma}\label{quotientingalltriang}
Given an element $\sigma \in S_{n}$, the canonical forms,  $\Omega(A_{n}^{T_{1},\sigma}) = \Omega(A_{n}^{T_{2},\sigma})$, (with an appropriate choice of orientation), for all triangulations $T_{1}$ and $ T_{2}$ of an $n$-gon.
\end{lemma}
\end{tcolorbox}
\begin{proof}
The canonical forms $\Omega(A_{n}^{T,\sigma})$ are the pull backs of the canonical form $\Omega(A_{n}^{T})$ on the embedding space via the diffeomorphism $f_{\sigma}^{-1}$. As the canonical form of $A_{n}^{T}$ in the embedding space is the same for all triangulations $T$, their pullbacks should also be the same. Therefore, $\Omega(A_{n}^{T_{1},\sigma}) = \Omega(A_{n}^{T_{2},\sigma})$, for all triangulations $T_{1}$ and $ T_{2}$ of an $n$-gon.
\end{proof}

The lemma \ref{quotientingalltriang} tells us that the sum over triangulations is redundant. We could choose to fix any triangulation and work with it. For concreteness, we fix reference $T$ to be, 
\begin{flalign}
T = \{ (1,3), \dots, (1,n-1) \}.
\end{flalign}

\begin{tcolorbox}[colback=black!5!white,colframe=black!75!black]
\begin{lemma}\label{quotientingbose}
Given a triangulation $T$, the canonical forms,  $\Omega(A_{n}^{T,\sigma_{1}})=\Omega(A_{n}^{T,\sigma_{2}})$, whenever $\sigma_{1}$ and $\sigma_{2}$ belong to the same right coset of $D_{n}\subset S_{n}$.
\end{lemma}
\end{tcolorbox}
\begin{proof}
    If $\sigma_{1}$ and $\sigma_{2}$ belong to the same right coset of $D_{n}\subset S_{n}$, then the linear isomorphisms $f_{\sigma_{1}}$ and $f_{\sigma_{1}}$ map the set of planar variables $\{x_{ij}\}$ to the same set. That is $\{f_{\sigma_{1}}(x_{ij})\} = \{f_{\sigma_{2}}(x_{ij})\}$. This means the poles of canonical forms $\Omega(A_{n}^{T,\sigma_{1}})$ and $\Omega(A_{n}^{T,\sigma_{2}})$ are the same. As the canonical form is determined, up to a sign, by the compatibility of poles, if the set of poles is the same, the canonical forms have to be equal.  
\end{proof}
The lemma \ref{quotientingbose} implies $\sum_{\sigma \in S_{n}} \Omega(A_{n}^{T,\sigma})= 2 n \sum_{[\sigma] \in \mathcal{G}_{n}} \Omega(A_{n}^{T,\sigma})$. Where $\mathcal{G}_{n}$ is the set of right cosets of $D_{n}$. That is, 
\begin{equation}
    \mathcal{G}_{n} = D_{n} \backslash  S_{n}. 
\end{equation}
And in the sum $\sum_{[\sigma] \in \mathcal{G}_{n}} \Omega(A_{n}^{T,\sigma})$, for each coset $[\sigma]$, we take some representative $\sigma \in [\sigma]$. With these redundancies taken care of, we can now write down the scattering amplitude of $\phi^{3}$ theory.
\begin{tcolorbox}[colback=black!5!white,colframe=black!75!black]
\begin{theorem}\label{phicubedamplitude}
The scattering amplitude of $\phi^{3}$ theory is given by 
\begin{equation}\label{phicubedamplitudeprime}
 \mathcal{M}_{n}^{\phi^{3}} (p_{1},\, \dots,\, p_{n} ) = \frac{1}{2^{n-3}} \sum_{[\sigma]\in \mathcal{G}_{n}} {\cal M}_{n}(\sigma^{\prime})
\end{equation}
for any triangulation $T$ and any $\sigma^{\prime}\, \in\, [\sigma]$. 
\end{theorem}
\end{tcolorbox}
\begin{proof}
    Let's first look at the sum $\sum_{\sigma\in S_{n}} \Omega(A_{n}^{T,\sigma})\vert_{A_{n}^{T,\sigma}}$. As discussed in section \ref{bose}, given any scattering channel $S$ of $\phi^{3}$ theory, there are $n 2^{n-2}$ elements of $S_{n}$ that take some planar diagram to $S$. Therefore
    \begin{equation}
        \sum_{\sigma\in S_{n}} {\cal M}_{n}(\sigma)\, = n\, 2^{n-2} \mathcal{M}_{n}^{\phi^{3}} .
    \end{equation}
    Further, as discussed in lemma \ref{quotientingbose} 
 \begin{flalign}\nonumber
\sum_{\sigma\in S_{n}} \Omega(A_{n}^{T,\sigma})\vert_{A_{n}^{T,\sigma}} = 2 n \sum_{[\sigma] \in \mathcal{G}_{n}} \Omega(A_{n}^{T,\sigma})\, =\, 2 n \sum_{[\sigma]\, \in\, \mathcal{G}_{n}}\, {\cal M}_{n}(\sigma^{\prime})\, \bigwedge_{(i,j)\, \in\, T}\, \dd s_{\sigma^{\prime}(i)\, \sigma(i+1) \dots \sigma^{\prime}(j-1)}
\end{flalign}     
 for any $\sigma^{\prime}\, \in\, [\sigma]$. 
 
 Therefore, 
\begin{equation}
 \mathcal{M}_{n}^{\phi^{3}} = \frac{1}{2^{n-3}} \sum_{[\sigma]\in \mathcal{G}_{n}} {\cal M}_{n}(\sigma^{\prime})
\end{equation}
    
\end{proof}


 This is one of the central results of the paper.\footnote{If we do not normalize the sum by $\frac{1}{2^{n-3}}$ then the result can also be  interpreted as follows: The planar scattering form which generates tree-level S matrix of bi-adjoint scalar theory with coupling $\lambda$ also generates the manifestly crossing symmetric S-matrix of a colorless $\phi^{3}$ theory, but with coupling $\sqrt{2}\, \lambda$.} It shows the precise manner in which the ABHY associahedron is the positive geometry for the \emph{manifestly} crossing-symmetric $\phi^{3}$ S-matrix. In fact, composing the $\sigma$-induced linear maps with a translation, 
\begin{flalign}
s_{\sigma(i)\, \dots,\, \sigma(j)}\, \rightarrow\, \overline{s}_{\sigma(i)\, \dots\, \sigma(j)}\, :=\, s_{\sigma(i)\, \dots,\, \sigma(j)}\, - m^{2}\, \forall\, (i,j)
\end{flalign}
 results in the n-point amplitude for massive $\phi^{3}$ theory, showing how ABHY associahedron is the ``amplituhedron" for scalar field S-matrix with cubic coupling and arbitrary mass. 
We end this section with an observation.
\begin{itemize}
\item Lemma \ref{ansigansigprime} implies that there exists a set of permutations $\{\sigma_{I}\}_{I=1}^{\vert {\cal G}_{n}\vert}$ where each $\sigma_{I}$ is in a different right coset of $D_{n}$ such that
$\oplus_{\sigma_{I}}\, A_{n-3}^{T, \sigma_{I}}$ is a positive  geometry in ${\cal K}_{n}$ with the corresponding canonical form,
\begin{flalign}
\Omega_{n-3}\, :=\, \sum_{\sigma_{I}}\, \Omega_{n-3}^{T, \sigma_{I}}\vert_{A_{n-3}^{T,\sigma_{I}}}
\end{flalign}
\end{itemize}
 \subsection{Deformed realisation for $n=4,\, n=5$}\label{drfn45}
In this sub-section, we provide a few explicit  examples of the deformed realizations $A_{n}^{T, \sigma}$ for $n\, =\, 4,\, 5$.
In the case of $n\, =\, 4$, ${\cal K}_{2}$ is coordinatized by $\{\, X_{13}\, =\, s,\, X_{24}\, =\,t\, \}$ and $\mathcal{G}_{4} = \{ [\textrm{id}] , [\sigma_{1}], [\sigma_{2}]\} $, with $\sigma_{1}\, =\, \bigl(\begin{smallmatrix}
    1 & 2 & 3 & 4 \\
    1 & 3 & 2 & 4
  \end{smallmatrix}\bigr),\, \sigma_{2}\, =\, \bigl(\begin{smallmatrix}
    1 & 2 & 3 & 4 \\
    1 & 3 & 4 & 2
  \end{smallmatrix}\bigr)\,$. It can be immediately checked that,
 \begin{flalign}
 \begin{array}{lll}
 A_{4}^{T,\, \sigma_{1}}\, =\, \{\, s_{13} \geq 0, X_{24}\, \geq\, 0\, \vert\, s = -\, c\, \}\\
 A_{4}^{T,\, \sigma_{2}}\, =\, \{ X_{13} \geq 0, s_{24}\, \geq\, 0\, \vert\, t = -\, c\, \}  
 \end{array}
 \end{flalign}
  The geometric realizations of deformed associahedra $A_{4}^{T, \sigma_{1}}$, $A_{4}^{T, \sigma_{1}}$ and the ABHY associahedra at 4 points are given in figure \ref{4ptnonplanarassociahedra}. 
 \begin{figure}
     \centering
     \includegraphics[scale=0.35]{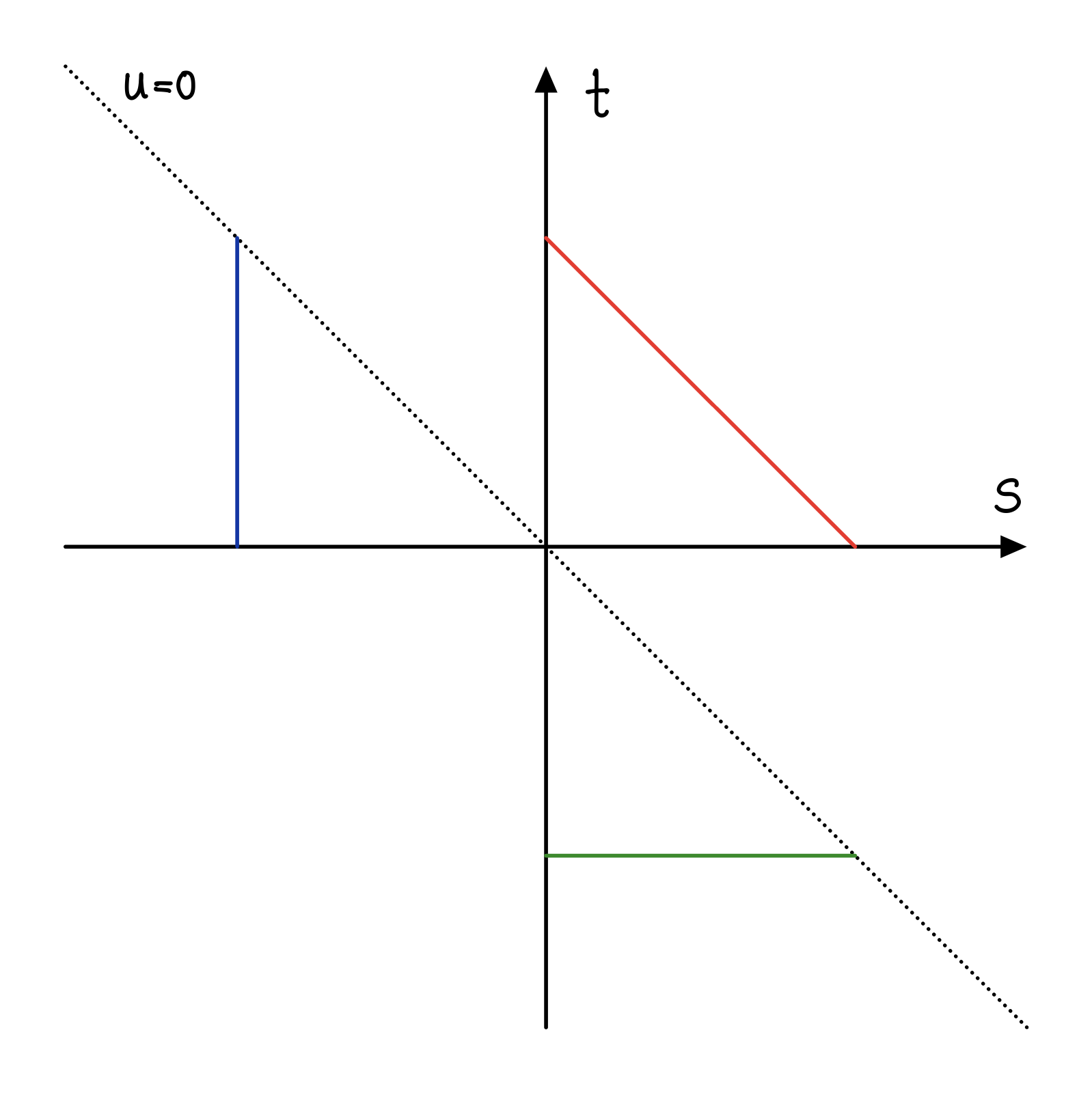}
     \caption{ABHY associahedron and its deformations for n = 4}
     \label{4ptnonplanarassociahedra}
 \end{figure}

The one forms on the three associahedra in ${\cal K}_{4}$ are obtained via pullback of $f_{\sigma}^{-1 \star}\, \Omega_{1}$ on the corresponding associahedra. 
\begin{flalign}
\begin{array}{lll}
\Omega(A_{4}^{T,e})\vert_{A_{4}^{T,e}} =\, (\, \frac{1}{s} + \frac{1}{t}\, )\, \dd s \\
\Omega(A_{4}^{T,e})\vert_{A_{4}^{T,\sigma_{1}}} =\, (\, \frac{1}{t} + \frac{1}{u}\, )\, \dd u\\
\Omega(A_{4}^{T,e})\vert_{A_{4}^{T,\sigma_{2}}} =\, (\, \frac{1}{s}\, +\, \frac{1}{u}\, )\, \dd s
\end{array}
\end{flalign}
Hence,
\begin{flalign}
\frac{1}{2}\, \sum_{[\sigma]\, \in\, {\cal G}_{4}}\, \Omega(A_{4}^{T,\sigma})\vert_{A_{4}^{T,\sigma}} =\, \frac{1}{s} + \frac{1}{t} + \frac{1}{u}\, =\, {\cal M}_{4}(\, p_{1},\, p_{2}, p_{3}  ,\, p_{4}\, )
\end{flalign}

As $\vert S_{5} \vert\, =\, 120$, and ${\cal G}_{5}\, =\, \{\, [e] , [\sigma_{1}] = \Big[ \bigl(\begin{smallmatrix}
    1 & 2 & 3 & 4 & 5 \\
    1 & 4 & 3 & 2 & 5
  \end{smallmatrix}\bigr) \Big] ,  [\sigma_{2}] =\, \Big[ \bigl(\begin{smallmatrix}
    1 & 2 & 3 & 4  & 5\\
    1 & 3 & 5 & 2 & 4 
  \end{smallmatrix}\bigr)\Big],  \ldots \}$, there are many deformed realizations of the ABHY associahedron $A_{5}^{T}$.  We will analyze two of them, which live in different quadrants of ${\cal K}_{5}$ and have an unequal number of non-planar channels.

Under the action of $f_{\sigma}\, \vert\, \sigma\, \in\, {\cal D}_{5}$ it is mapped to a convex realisation in ${\cal K}_{5}$. We first illustrate such a deformation with a couple of examples. 

For $\sigma_{1} = \bigl(\begin{smallmatrix}
    1 & 2 & 3 & 4 & 5\\
    1 & 5 & 2 & 3 & 4
  \end{smallmatrix}\bigr)$, 
\begin{flalign}
f_{\sigma_{1}}^{-1}\, (\, {\cal E}_{5}^{\geq\, 0}\, )\, =\, \textrm{span}(\, \{\, X_{25},\, X_{35},\, s_{25},\, s_{14},\, X_{24}\, \,  \}\, \geq\, 0\, )
\end{flalign}
In this case, $A_{5}^{T, \sigma_{1}}$ is a two-dimensional positive geometry defined by the  $X_{25},\, X_{35}\, \geq\, 0$ region of the 2-plane given by the following equations, 
\begin{flalign}\label{sigma12dd}
\begin{array}{lll}
X_{25}\, -\, X_{35}\, +\, s_{25}\, =\, c_{1}\\
X_{25}\, +\, s_{14}\, =\, c_{2}\\
X_{24}\, +\, X_{35}\, =\, c_{3},
\end{array}
\end{flalign}
where $c_{i}$ are arbitrary positive constants. The five co-dimension-one boundaries of this polytope are located at 
\begin{flalign}\nonumber
\{\, X_{25}, X_{35}, s_{25},\, s_{14},\, X_{24}\, \}\, \rightarrow\, 0
\end{flalign}
The two sided bounds imposed by eqn.(\ref{sigma12dd}) on these kinematic variables imply that, 
\begin{flalign}
\begin{array}{lll}
X_{13} = - c^{\prime}\\
X_{14}\, =\, c_{2}\, -\, X_{24}
\end{array}
\end{flalign}
which shows how $A_{5}^{T, \sigma}$ has no intersection with  the ABHY associahedron $A_{5}^{T}$. 

In the second example $\sigma_{2} = \bigl(\begin{smallmatrix}
    1 & 2 & 3 & 4 & 5\\
    1 & 4 & 2 & 5 & 3
  \end{smallmatrix}\bigr)$,
\begin{flalign}
f_{\sigma_{2}}^{-1}\, (\, {\cal E}_{5}^{\geq\, 0}\, )\, =\, \textrm{span}(\, \{\, s_{14},\, s_{35},\, s_{24},\, s_{25},\, s_{13}\, \}\, \geq\, 0\, )
\end{flalign}
In this case, $A_{5}^{T, \sigma_{2}}$ is a two-dimensional positive geometry  defined by $s_{14}, s_{35}\, \geq\, 0$ region in the 2-plane defined by the equations, 
\begin{flalign}\label{sigma12ddprime}
\begin{array}{lll}
s_{14}\, +\, s_{24}\, -\, s_{35}\, =\, d_{1}\, =\, -\, X_{13}\\
s_{35}\, +\, s_{25}\, =\, d_{2}\, =\, -\, ( X_{25}\, +\, X_{14}\, )\\
s_{14}\, +\, X_{13}\, =\, d_{3}\, =\, -\, (X_{13}\, +\, X_{25}\, )
\end{array}
\end{flalign}
where, as before $d_{1}, d_{2}, d_{3}$ are arbitrary positive constants. None of  the co-dimension one faces of $A_{5}^{T,\, \sigma_{2}}$ correspond to planar poles of the 5-point amplitude. The above equations imply that  $A_{5}^{T, \sigma_{2}}$ is the two-dimensional pentagon which is the intersection of hyper-planes 
\begin{flalign}
X_{13}\, =\, -d_{1}, X_{25}\, =\, -\, d_{3}\, -\, d_{1}, X_{14}\, =\, -\, \sum_{i}\, d_{i}
\end{flalign}
with the cone, 
\begin{flalign}
X_{24}\, +\, 2\, d_{1}\, + 2\, d_{3}\, + d_{2}\, \geq\, 0\, \textrm{and}\, X_{35}\, \leq\, d_{1}\, +\, d_{2}.
\end{flalign}
Once again, it has no intersection with the ABHY associahedron, and in fact, none of the three associahedra $A_{n-3}^{T},\, A_{n-3}^{T\, \sigma_{1}},\, A_{n-3}^{T\, \sigma_{2}}$ intersect each other as they all lie in distinct hyper-planes in ${\cal K}_{n}$. 

We end this section with a few remarks.
\begin{itemize}
\item In \cite{Jagadale:2022rbl}, it was proved how diagonal linear maps of the form $x_{ij}\, =\, \alpha_{ij}\, X_{ij}\, -\, m_{ij}^{2}\,$  for all $(i,j)$ with $ 1\, \leq\, i\, \leq\, j\, \leq\, n$ generates deformed realization of ABHY associahedron in ${\cal K}_{n}^{\geq\, 0}$ which turns out to be positive geometry of color ordered S-matrix of cubic scalar non-derivative interactions where the number of scalar fields, their mass parameters and the strength of various cubic couplings is contained in a family of parameters $(\, \alpha_{ij},\, m_{ij}\, )$. We now show that the composition of such diagonal maps with ${\cal G}_{n}$ induced diffeomorphism leads us to positive geometries for non-planar scalar S matrix with cubic interactions between distinct scalar fields with arbitrary masses.  
\item ABHY associahedron is thus a universal polytope whose simplest possible avatars (obtained simply by linear mappings of embedding space ${\cal E}_{n}$ in the kinematic space) leads to an entire spectrum of tree-level S-matrix with scalar particles.
\end{itemize}
\section{CHY formula for the  $\phi^{3}$ S-matrix without color.}\label{chynp}
The discovery of the diffeomorphic avatars of the ABHY associahedron in the kinematic space has intriguing consequences for  the worldsheet formulation of tree-level S matrix given by Cachazo, He, and Yuan (CHY) in \cite{Cachazo:2013iea}. In essence, building on the seminal result in \cite{Arkani-Hamed:2017mur}, a point of view advocated in \cite{Jagadale:2022rbl} was that the scattering equations should be considered as  diffeomorphism between the (real section) of the compactified moduli space $\overline{{\cal M}}_{0,n}({\bf R})$ and the ABHY associahedron $A_{n-3}^{T}$ in ${\cal E}_{n}^{\geq\, 0}$. That is, given the $\{y_{ij}\}$ coordinates introduced in equation (\eqref{yijxij}),  we consider the ``embedding space" scattering equations,
\begin{flalign}
\sum_{i\, \neq\, j}\, \frac{y_{ij}}{z_{ij}}\, =\, 0.
\end{flalign}
which define a map from the worldsheet to ${\cal E}_{n}$. Restriction of the scattering equations to $A_{n-3}^{T}$ via 
\begin{flalign}
{\cal F} = 0\, \equiv\, \{\, y_{ij}\, =\, -\, c_{ij}\, \forall\, (i,j)\, \notin\, \{\, (2,n),\, (3,n),\, \dots,\, (n-2,n)\, \}\, \}
\end{flalign}
defines a diffeomorphism between the worldsheet associahedron  $A_{n-3}^{\textrm{ws}}\, :=\, \overline{{\cal M}}_{0,n}({\bf R})$ and the ABHY associahedron in ${\cal E}_{n}^{\geq\, 0}$.

Hence composing the diffeomorphism induced by scattering equations with linear map parametrized by $[\sigma]\, \in\, {\cal G}_{n}$ gives us a  (family of) maps between $A_{n-3}^{\textrm{ws}}$ and the kinematic space associahedra $A_{n-3}^{T\, [\sigma]}$.  We call these diffeomorphisms $[\sigma]$-deformed scattering equations obtained by the identifying 
\begin{flalign}
y_{ij}\, =\, s_{\sigma(i)\sigma(j)} 
\end{flalign}
and imposing
\begin{flalign}\label{fsigma}
{\cal F}([\sigma]) = 0\, \equiv\, \{\, s_{\sigma(i)\sigma(j)}\, =\, -\, c_{ij}\, \forall\, (i,j)\, \notin\, \{\, (2,n),\, (3,n),\, \dots,\, (n-2,n)\, \}\, \}
\end{flalign}
The canonical form on $A_{n-3}^{\textrm{ws}}$ is the Parke Taylor form defined as,
\begin{flalign}
\omega_{\textrm{ws}}(z_{1},\, \dots,\, z_{n})\, =\, \frac{\dd^{n-3}z}{z_{12}\, z_{23}\, \dots\, z_{n1}}
\end{flalign}
where $z_{ij}\, :=\, z_{j} - z_{i}$. 

The CHY formula for massless $\phi^{3}$ S matrix can now be written down immediately. As it is simply a weighted sum over the Park-Taylor form evaluated on the solution of the $[\sigma]$-deformed scattering equations.
\begin{align}
   & {\cal M}_{n}(p_{1}, \dots\, p_{n})\, \endline & =\, \frac{1}{2^{n-3}}\, \sum_{[\sigma]\, \in\, {\cal G}_{n}}\, \int_{\textrm{ws}}\, \omega_{WS}\, \left[\, \prod_{i=1}^{n}\, \frac{1}{z_{i} - z_{i+1}}\, \prod_{k}^{\prime}\, \delta\left(\, \sum_{m\, \neq\, k}\, \frac{s_{\sigma(m)\sigma(k)}}{z_{mk}}\, \right)\vert_{{\cal F}([\sigma])\, =\, 0}\, \right]
\end{align}
where $\prod_{k}^{\prime}$ denotes a product over the $n-3$ punctures after removing the $SL(2,R)$ redundancy and the scattering equations are evaluated on ${\cal F}([\sigma])\, =\, 0$ in eqn.(\ref{fsigma}). We note that this formula could have been obtained directly  (i.e., without exploiting the relationship between the CHY formula and canonical form of the associahedron) from the observations made in section \ref{bose}.

\section{Deformed Stokes Polytopes in the kinematic space}\label{phi4}
The positive geometries for planar scattering amplitudes in theory with polynomial interactions belong to the family of polytopes called accordiohedra \cite{Aneesh:2019ddi, Jagadale:2021iab}. An accordiohedron polytope is parametrized by a dissection of a planar $n$-gon where the dissection is coarser than the complete triangulation. The geometric realizations of accordiohedra descend from the geometric realization of the associahedron by projecting the associahedron equation defined with respect to, say, $T$ to the coarser dissection.  The relevance of accordiohedron to the S-matrix program was discovered in \cite{Banerjee:2018tun, Aneesh:2019cvt} where it was shown that the accordiohedron parametrized by quadrangulations constitute positive geometry of color-ordered S-matrix with quartic scalar interactions. 

Any quadrangulation $Q$ of an $n$-gon generates a Stokes polytope.  Although the precise combinatorial definition is rather involved, the construction can be understood as follows: $Q$ divides the $n$-gon into a collection of cells (quadrilaterals).  Now consider a dual polygon whose vertices correspond to the edges of the original $n$-gon. Any quadrangulation $Q^{\prime}$ of the dual $n$-gon is said to be compatible with $Q$ if any chord in $Q^{\prime}$ only enters and exits a given cell of the $n$-gon via adjacent edges. A Set of all such $Q^{\prime}$s related to each other via mutations generate the Stokes polytope ${\cal AC}_{Q}$, which is a closed and convex positive geometry.  We refer the reader to \cite{dmtcs:2572, Manneville2019GeometricRO} for more details on accordiohedron and Stokes Polytope.

To obtain a geometric realization of the Stokes polytope $\mathcal{AC}_{Q}$ associated with a quadrangulation 
\begin{displaymath}
Q= \{(i_{1}j_{1}), \ldots (i_{\frac{n-4}{2}} j_{\frac{n-4}{2}}) \}, 
\end{displaymath}
we consider a triangulation $T= \{ (i_{1} j_{1}), \ldots (i_{n-3} j_{n-3}) \}$ such that $Q\subset T$. The geometric realization of $\mathcal{AC}_{\{Q,T\}}$ is obtained simply via projecting $A_{n-3}^{T}$ on to the subspace spanned by the variables 
\begin{displaymath}
\{ X_{i_{1}j_{1}} , \ldots X_{i_{\frac{n-4}{2}} j_{\frac{n-4}{2}}} \}
\end{displaymath}
Given a quadrangulation $Q$, it fixes a unique planar  $\frac{n-4}{2}$ form in ${\cal K}_{n}$ as follows. 
This immediately implies that the restriction of the $\frac{n-4}{2}$ planar scattering form $\Omega_{n}^{Q}$ to the ABHY associahedron $A_{n-3}^{T}$ generates the planar quartic scalar amplitudes ${\cal M}^{\phi^{4}}_{n}(p_{1}\, \dots\, p_{n})$ 
\cite{Aneesh:2019cvt}.
\begin{flalign}
\Omega_{\frac{n-4}{2}}^{Q(T)}\, \vert_{A_{n-3}^{T}}\, =\, {\cal M}^{\phi^{4}}_{n}(p_{1}\, \dots\, p_{n})\, \bigwedge_{(i,j)\, \in\, Q(T)}\, \dd x_{ij}
\end{flalign}
Given any $\sigma\, \in\, S_{n}$ we can once again consider pull back of the projective $\frac{n-4}{2}$ form on ABHY associahedron.\footnote{If $\sigma\, \in\, S_{n}$ then the deformed associahedron is of ABHY type, but realized in a different positive quadrant of the kinematic space.} 
\begin{flalign}
f_{\sigma}^{-1 \star}\, \cdot\, \Omega_{Q(T)}\, \vert_{(A_{n-3}^{T, \sigma})}\, =\, \sum_{Q^{\prime}}\, [\, \prod_{(m,n)\, \in\, Q^{\prime}}\, \frac{1}{s_{\sigma(m)\, \cdot\, \sigma(n)-1}}\, ]\, \bigwedge_{(i,j)\, \in\, Q}\, \dd s_{\sigma(i)\, \dots\, \sigma(j)-1}
\end{flalign}
where the sum is over all the quadrangulations compatible with the reference quadrangulation. 

The color-ordered (planar) tree-level S matrix of $\phi^{4}$ theory is a weighted sum over canonical forms associated with Stokes Polytopes ${\cal AC}_{Q}$. Every Stokes polytope is realized in the linearity space $\{\, X_{ij}\, \geq\, 0\, \vert\, (i,j)\, \in\, Q\, \}$ inside ${\cal K}_{n}^{\geq\, 0}$.  As we now argue, this result can also be understood via the ideas introduced in this paper. 

Consider the dihedral group ${\cal D}_{n}\, \subset\, S_{n}$ and label all the permutations in ${\cal D}_{n}$ as $\tilde{\sigma}$. Let $\{Q_{I}\}$ be the set of all the primitives (topologically inequivalent quadrangulations of an $n$-gon as defined in \cite{Raman:2019utu}.

 Given a $\tilde{\sigma}\, \in\, {\cal D}_{n}$ and a reference quadrangulation $Q_{I}$, let ${\cal M}_{n}(\tilde{\sigma}, Q_{I})$ be defined via the following equation. 
\begin{flalign}
\sum_{\tilde{\sigma}\, \in\, {\cal D}_{n}}\, (f_{\tilde{\sigma}}^{-1})^{\star}\, \cdot\, \Omega_{Q_{I}}^{\frac{n}{2}\, -\, 2}\vert_{f_{\tilde{\sigma}}(\, A_{n-3}^{Q_{I}\, \subset\, T_{I}}\, )}\, =:\, \sum_{\tilde{\sigma}\, \in\, {\cal D}_{n}}\, {\cal M}_{n}(\tilde{\sigma}, Q_{I})\, \bigwedge_{(1,m)\, \in\, Q_{I}}\, d s_{1\tilde{\sigma}(2)\, \dots\, \tilde{\sigma}(m-1)}
\end{flalign}
 Then the color-ordered $\phi^{4}$ amplitude, ${\cal M}_{n}^{\textrm{co}}(p_{1},\, \dots,\, p_{n})$ is given by the formula
\begin{flalign}\label{mncophi4}
{\cal M}_{n}^{\textrm{co}}(p_{1}\, \dots\, p_{n})\, =\, \sum_{I}\, \alpha_{I}\, \sum_{\tilde{\sigma}\, \in\, {\cal D}_{n}}\, {\cal M}_{n}(\tilde{\sigma},Q_{I})
\end{flalign}
where $\alpha_{I}$ are the weights associated to the primitives $\{\, \tilde{\sigma}\, \cdot\, Q_{I}\, \vert\, \tilde{\sigma}\, \in\, {\cal C}_{n}\, \}$. These weights were analyzed in \cite{Raman:2019utu, Kojima:2020tox, John:2020jww, Aneesh:2019ddi}. A general discussion on the structure of weight and the formula for computing them for a generic accordiohedron can be found in \cite{Jagadale:2021iab}.

Hence the color-ordered tree-level amplitude in massless $\phi^{4}$ theory can also be understood as a weighted sum  over ``deformed" realizations of the accordiohedra ${\cal AC}_{Q}$ where the deformations are simply linear maps from ${\cal E}_{n}$ to ${\cal K}_{n}$ and are parametrized by the dihedral group over $n$-elements. 
 \subsection{Towards S-matrix of $\phi^{4}$ theory without color}
 As the boundaries of the deformed realization of associahedron contain all the poles of a tree-level S matrix, it is rather tempting to speculate if the evaluation of the lower rank projective forms on these realizations will generate amplitudes of scalar field theories with generic non-derivative interactions. We now argue that this is indeed the case if we consider the projective $\frac{n-4}{2}$ forms parametrized by quadrangulations as it generates tree-level S-matrix of $\phi^{4}$ theory without color. 
 
We recall once again the notion of a primitive quadrangulation: A primitive is defined to be an equivalence class over those quadrangulations which can be mapped to each other by an element of $C_{n}$. For the purpose of this section it will be useful to introduce a more ``coarser" classification over the set of quadrangulations which we refer to as a labeled graph, $\Gamma_{q}$.
 
Consider a graph with four classes of nodes.\footnote{The labeled graphs defined here are isomorphic to the dissection quivers associated with quadrangulations, \cite{Padrol2019AssociahedraFF}. However, for the purpose of this discussion, the labeled graph representation is more appropriate.} We will label these 4 classes as  $\{$ r, b, v, and g$\}$. 
\begin{flalign}
\begin{array}{lll}
r\, \sim\, \textrm{a node with three external legs}\\
b\, \sim\, \textrm{a node with two external legs}\\
v\, \sim\, \textrm{a node with one external leg}\\
g\, \sim\, \textrm{a node with no external leg}
\end{array}
\end{flalign}
Every quadrangulation of an $n$-gon generates a labeled graph. However, this correspondence is many-to-one, and thus each $\Gamma_{q}$ defines an equivalence class of quadrangulations. The following lemma shows how $\Gamma_{q}$ generates a coarser classification scheme for quadrangulations as compared to the notion of primitives
\begin{lemma}
Consider an $n$-gon  whose vertices are labelled clockwise as $\{\, 1,\, \dots,\, n\, \}$. Let $Q_{1},\, Q_{2}$ be two distinct quadrangulations that belong to two distinct primitives but belong to the same equivalence class labeled by a labeled graph $\Gamma_{q}$. Then there exists a $\sigma_{1,2}\, \in\, S_{n}$ under whose action $Q_{1}$ is mapped to $Q_{2}$. 
That is, if $Q_{1}$ is the quadrangulation of $\{\, 1,\, \dots,\, n\, \}$, then keeping $Q_{1}$ fixed and permuting external vertices by $\sigma_{1,2}$ is equivalent to keeping the $n$-gon fixed and changing $Q_{1}$ to $Q_{2}$. 
\end{lemma}
\begin{proof}
Let $\{\, N_{1},\, \dots,\, N_{\frac{n-4}{2}}\, \}$ be the set of nodes of $\Gamma_{q}$ where each $N_{i}\, \in\, \{\, \textrm{r, b, v, g}\, \}$.  Let $N_{1}$ be labelled by external vertices $\{\, 1, 2,\, \dots,\, k\, \}\, \vert\, 0\, \leq\, k\, \leq\, 3$.  
Let the same node be labelled by external vertices $\{\, i_{1}, i_{2}, i_{k}\, \}$ in $Q_{2}$. Consider a permutation $\sigma_{N_{1}}$ that maps $ (2, \dots,\, k)$ to $(i_{2}, \dots,\, i_{k} )$ while keeping all the other vertices same.  Let $\, \{\, N_{1},\, \dots,\, N_{\frac{n-4}{2}}\, \}$ be the set of nodes of  the labelled graph $q$. Then the desired element of $S_{n}$ is $( \sigma_{N_{1}}\, \circ\, \sigma_{N_{2}}\, \circ\, \dots\, \circ\, \sigma_{N_{\frac{n-4}{2}}}\, )$. 
\end{proof} 
The family of labeled graphs parametrized by $n$ has the following properties.
\begin{enumerate}
\item For $n\, \leq\, 8$ there is a unique $\Gamma_{q}$ The $n = 8$ quiver is simply $\Gamma_{q} = \{\, r, b, r\, \}$.
\item If $n\, =\, 10$, there are two labelled graphs,
\begin{flalign}\label{cqn10}
\begin{array}{lll}
\Gamma_{q_{1}}\, =\, \{\, r, b, b, r\, \}\\
\Gamma_{q_{2}} = \{\, r, r, r, v\, \}.
\end{array}
\end{flalign}
Note that $Q_{1}\, =\, \{\, (1,4), (1,6), (1,8)\, \}$ for which ${\cal AC}_{Q_{1}}$ is a 3-dimensional associahedron belongs to $\Gamma_{q_{1}}$.
\item There are two vertices in $V({\cal AC}_{Q_{1}})$ that belongs to $\Gamma_{q_{2}}$ whereas the other 12 belong to $q_{1}$. That is for $n\, \leq\, 10$, the set of quadrangulations in $V({\cal AC}_{Q=\{\, (1,4), \dots, (1,n-2)\, \}})$ exhaust all the labelled graphs. \footnote{For $n\, =\, 6,8$ this can be checked readily. For $n=10$, we refer the reader to the appendix of  \cite{Banerjee:2018tun}.}  
\item If $n\, \geq\, 12$ then there is no quadrangulation $Q$ for which the set  $V({\cal AC}_{Q})$ spans over all the corresponding labelled graphs.\footnote{This can be argued as follows. For $n\, =\, 12$, we can verify this directly. Let the opposite be true if $n\, \geq\, 14$, but this leads to contradiction as there is no boundary of ${\cal AC}_{Q}$  which corresponds to 4 dimensional Stokes polytope and whose set of vertices exhaust all the labeled graphs in $n = 12$ case.}
\end{enumerate}
We will now try to apply the ideas of the previous section to relate  non-planar $n$ point scalar amplitudes in $\phi^{4}$ theory with projective lower forms $\Omega_{Q}$. Let, 
\begin{flalign}
Q\, =\, \{\, (1,4),\, (1,6),\, \dots,\, (1,n-2)\, \}
\end{flalign}
${\cal AC}_{Q}$ is an $\frac{n-4}{2}$ dimensional associahedron with $C_{\frac{n-2}{2}}$ number of vertices.  

Let us first define a $\sigma$-deformed Stokes polytope ${\cal AC }_{Q}^{\sigma}$ : Given a quadrangulation $Q$ of an $n$-gon with a clockwise ordering of vertices and the corresponding ${\cal AC}_{Q}$, ${\cal AC}_{Q}^{\sigma}$ is the Stokes polytope under the action of $\sigma$ on the vertices of the $n$-gon. 

From the comments made above, we can deduce the following surjection.
\begin{flalign}\label{overlinevq}
\overline{V}_{Q}\, :=\cup_{\sigma\, \in\, S_{n\, \leq\, 10}}& V({\cal AC}_{Q}^{\sigma})\nonumber\\
&\downarrow\nonumber\\
\textrm{Set of all quadrangulations of}& (\sigma(1),\, \dots,\, \sigma(n\, \leq\, 10))\textrm{-gon $\forall\, \sigma\, \in\, S_{n\, \leq\, 10}$.}
\end{flalign}
In the case of associahedron, we had the following map for a fixed triangulation $T$.
\begin{flalign}
\cup_{[\sigma]\, \in\, {\cal G}_{n}}& V((A_{n-3}^{T, \sigma}))\nonumber\\
&\downarrow\nonumber\\
\textrm{Set of all triangulations}&\, \textrm{with all possible ordering of vertices.}
\end{flalign}
In fact, as we showed, this map is not only surjective but $2^{n-3}\, \rightarrow\, 1$ because each  vertex (complete triangulation) occurs precisely $2^{n-3}$ times in $\cup_{[\sigma]\, \in\, {\cal G}_{n}}\, V(A_{n-3}^{T, \sigma})$. 

However, in the case of Stokes polytope, the combinatorics is far more intricate. 
\begin{itemize}
\item  $\overline{V}_{Q}$ does not exhaust all the channels of tree-level amplitude $n\, \geq\, 12$.  For example, in the case of $n\, =\, 12$, direct inspection shows that any channel that can be labeled by the following Mandelstam invariants
\begin{displaymath}
\{\, s_{abc}, s_{mnk}, s_{pqr}, s_{wxy}\,  \}\, \textrm{with}\, a\, \neq\, b\, \dots\, \neq\, y 
\end{displaymath}
can not be a vertex belonging to the set $\overline{V}_{Q}$ defined in equation (\ref{overlinevq}).\footnote{In $n\, =\, 12$ example, only those graphs whose topology generates channel of this type is not contained in $\overline{V}$. However as we go to higher $n$, any channel which corresponds to poles of the form $\{\, s_{i_{1}j_{1}k_{1}},\, \dots,\, s_{i_{K}j_{K}k_{K}}\, \vert\, K\, >\, \frac{n}{4}\, \}$ are not contained in $\overline{V}$.}  
\item As we will see, for $n\, \leq\, 10$ the map defined in eqn.(\ref{overlinevq}) is many-one for $Q = \{ (1,4), (1,6), \dots, (1,n-2)\, \}$. However, for $n = 10$, all the quadrangulations in the range of the map do not occur with equal multiplicity leading to a rather subtle connection with the 10-point amplitude in $\phi^{4}$ theory.
\end{itemize}
The first issue can be resolved as follows.  We can simply consider a set of quadrangulations, each one selected from one labeled graph, 
\begin{flalign}\label{overlineV}
\overline{V}\, =\, \cup_{Q^{\prime}\in\, \Gamma_{q_{1}}\, \cup\, \Gamma_{q_{2}}\, \cup\, \dots\, \cup\, \Gamma_{q_{M}}}\, \overline{V}_{Q^{\prime}}
\end{flalign}
where the sum is taken over a set of labeled graphs in the following way. We can start with 
\begin{flalign}
Q_{1}\, =\, \{\, (1,4),\, \dots,\, (1,n-2)\, \}
\end{flalign}
Suppose $( v_{1}, \dots, v_{M})\, \subset\, V({\cal AC}_{Q})$ do not belong to $\Gamma_{q_{1}}$ and let
\begin{flalign}\nonumber
v_{1}\, \in\, \Gamma_{q_{2}}
\end{flalign}
We then include $\Gamma_{q_{2}}$ in the union. Let $v_{l\, \geq\, 3}\, \notin V({\cal AC}_{Q_{2}})$. Then we include the labeled graph associated with $v_{l}$ in the union and so on till we include all elements of $(v_{1},\, \dots,\, v_{M})$. 

 As an example, consider $n\, =\, 12$. We can choose 
 \begin{flalign}\nonumber
 \overline{V} = V_{Q}\, \cup\, V_{Q^{\prime}}\ \textrm{with}
 \end{flalign}
\begin{flalign}
Q^{\prime}\, =\, \{\, (1,4), (4,7), (7,10), (10,1)\, \}
\end{flalign}
Systematic classification of the minimal set of labeled graphs is an intriguing combinatorics problem. Although it is beyond the scope of this paper. And finally, even if we generate such a vertex set $\overline{V}$, not all the quadrangulations (or, more precisely, labeled graphs) will occur with the same multiplicity. We will see how to solve this problem in the $n = 10$ case and briefly comment on it for general $n$ at the end of this section.\footnote{Essentially, we believe that a rather brute-force way to solve this problem is to consider a weighted sum of projective lower forms over all the labeled graphs as opposed to a minimal set of labeled graphs whose vertex set exhausts all possible channels of an $n$-pt. amplitude in $\phi^{4}$ theory. This belief is reflected in our final conjectured formula proposed in  eqn.\ref{cpfan}} 

We now focus on the three Stokes polytopes of dimensions $\leq\, 3$. By explicit computation, we show that the (manifestly crossing symmetric) $\phi^{4}$ amplitude is a sum over the push-forward of the projective $\frac{n-4}{2}$ forms $\Omega_{Q}$ evaluated on the family of deformed associahedra, $\{\, A_{n-3}^{T, \sigma}\, \vert\, \sigma\, \in\, S_{n}\, \}$.  

Following comment is in order. 
\begin{itemize}
\item Our analysis can also be interpreted  without taking recourse to associahedron and working solely with (convex realizations) of Stokes Polytopes. In other words, given a reference quadrangulation $Q$ and the corresponding combinatorial polytope ${\cal AC}_{Q}$, the convex realisation of ${\cal AC}_{Q}$ in ${\cal E}_{n}^{\geq\, 0}$ is obtained by solving the system of equations \cite{Padrol2019AssociahedraFF}
\begin{flalign}
s_{ij}\, =\, -\, c_{ij}\, \forall\, (i,j)\, \notin\, \{\, (2,n), (4,n), \dots, (n-3,n)\, \}
\end{flalign}
We can now consider deformed realizations of ${\cal AC}_{Q}$ in ${\cal K}_{n}$ by using $f_{\sigma}\, \vert\, \sigma\, \in\, S_{n}$ and as in the tri-valent case, analyze the sum over canonical forms associated to all the deformed realizations. The two approaches are equivalent.  As in \cite{Aneesh:2019cvt} however, we will only use ${\cal AC}_{Q}$ as a combinatorial polytope used to define planar scattering form in ${\cal E}_{n}$. This perspective places the ABHY associahedron at the heart of the landscape of tree-level  scalar amplitudes. 
\end{itemize}
\subsection{$\phi^{4}$ amplitudes for $n\, \in\, \{\, 6,\, 8\, \}$.}
We will first illustrate this result for Six and Eight pt. amplitudes before addressing the generic $n$ point S-matrix.

In the six-point case, the situation is rather straightforward. The Stokes polytope is the one-dimensional associahedron. As we sum over all the projective one-form $\Omega_{Q = (1,4)}$ which are push-forwarded onto $A_{3}^{T, \sigma}$ (where $T$ is any triangulation that contrains the chord $(1,4)$), we get $2\, \cdot\, 6!$ terms in all. It can now be readily checked that if we define ${\cal M}_{6}(\sigma)$ via the following formula.

Henceforth we denote the pull back of the $\frac{n-4}{2}$ d-$\ln$ form as, 
\begin{flalign}
(f_{\sigma}^{-1})^{\star}\, \circ\, \Omega_{Q}\, :=\, \Omega_{Q}^{\sigma}
\end{flalign}

\begin{flalign}
\frac{1}{2\, \cdot\, 4!}\, \sum_{\sigma\, \in\, S_{6}}\, \Omega^{\sigma}_{Q(T)}\vert_{(A_{3}^{\sigma, T})}\, =:\, \sum_{\sigma\, \in\, S_{6}}\, {\cal M}_{6}(\sigma)\, \bigwedge_{(i,j)\, \in\, Q(T)}\, \dd s_{\sigma(1)\sigma(2)\sigma(3)}.
\end{flalign}
Then 
\begin{flalign}
{\cal M}_{6}(p_{1},\, \dots,\, p_{6})\, =\, \sum_{S_{6}}\, {\cal M}_{6}(\sigma)
\end{flalign}
The normalisation factor $\frac{1}{2\, \cdot\, 4 !}$ is simply the multiplicity with which all the poles $\frac{1}{s_{ijk}}$ as we sum over all the permutations.   We note that 
\begin{flalign}
\frac{\vert\, \cup_{\sigma\, \in\, S_{8}}\, V({\cal AC}^{\sigma}_{Q(T)})\, \vert}{\textrm{Multiplicity}}\, =\, 10 
\end{flalign}
which is precisely the number of Feynman diagrams in 6-point amplitude. 

In the $n\, =\, 8$ case there are two primitives. One in which the two chords of a quadrangulation intersect in a common vertex and the other in which two chords are parallel. The resulting Stokes polytopes are a two-dimensional associahedron and a square, respectively. That is, for $Q_{1}\, =\, \{\, (1,4),\, (5,8)\, \}$, ${\cal AC}_{Q_{1}}$ is a square. On the other hand for  $Q_{2}\, =\, \{\, (1,4),\, (1,6)\, \}$, ${\cal AC}_{Q_{2}}$ is $A_{2}$. However $Q_{1}\, \sim\, Q_{2}$ under the action of $S_{8}$ and hence there is a unique labelled graph $q$ for $n=8$. 

Consider 
\begin{flalign}
\begin{array}{lll}
Q_{1}\, =\, (14,16)\\
Q_{2}\, =\, (14,58)
\end{array}
\end{flalign}
and let, 
\begin{flalign}
\sigma : (1,\, \dots,\, 8)\, \rightarrow\, (\, 1, 2, 3, 4, 8, 5, 6, 7\, ),
\end{flalign}
We see that $\sigma\, \cdot Q_{1}\, =\, Q_{2}$. 

One can first determine the multiplicity of any codimension three face (in the set of all the $5$ dimensional deformed associahedra) which corresponds to a pole in $\phi^{4}$ amplitude. As the action of $S_{8}$ on the set of all such faces is transitive we have,
\begin{flalign}
\textrm{Multiplicity of}\, (s_{\sigma(1),\, \dots,\, \sigma(3)},\, s_{\sigma(1),\, \dots\, \sigma(5)} )\, =\, 3 !\, \cdot\, 3 !\, \cdot 2
\end{flalign}
Thus the (normalised) push-forward of $\Omega_{Q}$ on the deformed realizations is given by, 
\begin{flalign}
{\cal M}_{8}(\sigma)\,  \wedge_{(i,j)\, \in\, Q = (14,16)}\, d s_{\sigma(i)\sigma(i+1)\, \dots\, \sigma(j-1)}\, :=\, \frac{1}{3 ! 3! 2}\, 
\Omega^{\sigma}_{Q=(1,4), (1,6)}\vert_{(A_{5}^{\sigma, T})}
\end{flalign}
\begin{flalign}
{\cal M}_{8}(p_{1},\, \dots,\, p_{8})\, =\, \sum_{\sigma\, \in\, S_{8}}\, {\cal M}_{8}(\sigma)
\end{flalign}
Once again, we see that,
\begin{flalign}
\frac{\vert\, \cup_{\sigma}\, V({\cal AC}^{\sigma}_{Q(T)})\, \vert}{3 ! 3 ! 2}\, =\, \frac{5\, \cdot 8!}{3!3!2}\, =\, 280\, =\, \vert\,  \textrm{Feynman diagrams}\, \vert
\end{flalign}
\subsection{Higher point amplitudes as forms.}
For $n\, =\, 10$, the situation appears to be far more intricate for the following reason. Let
\begin{flalign}
Q_{1}\, =\, \{\, (1,4), (1,6), (1,8)\, \}
\end{flalign}
Then $\cup_{\sigma\, \in\, S_{10}}\, V({\cal AC}^{\sigma}_{Q_{1}})$ contains  the entire set of quadrangulations of the $10$-gon with all possible ordering of the vertices. However, the multiplicity of various vertices under the action of all permutations is not the same.

Consider, two vertices $v_{1},\, v_{2}$ of the  ABHY realisation of $A_{7}^{T\, =\, \{ (1,3),\, \dots,\, (1,9)\, \}}$ on which the projective three-form $\Omega_{Q_{1}}$ has poles. 
\begin{flalign}
\begin{array}{lll}
v_{1}\, =\, \{\, (1,4), (1,6), (1,8)\, \}\, \rightarrow\, \{\, s_{123}, s_{12345}, s_{8,9,10}\, \}\\
v_{2}\, =\, \{\, (1,4), (4,7), (7,10)\, \}\, \rightarrow\, \{\, s_{123}, s_{456},\, s_{789}\, \}
\end{array}
\end{flalign}
As we show below, $v_{1}$ and $v_{2}$ do not occur with equal multiplicity  in $\cup_{\sigma\, \in\, S_{10}}\, V({\cal A}_{7}^{T, \sigma} ) )$. Thus all poles do not contribute equally in the sum over $\Omega_{Q_{1}}^{\sigma}\vert_{A_{n-3}^{T, \sigma}}$ and as a result, such a sum is not an amplitude of any theory.

That is, even though $\overline{V}_{Q_{1}}$ contains all the quadrangulations dual to all possible Feynman diagrams in $\phi^{4}$ theory, to ensure  an equal multiplicity of all the vertices there must exist a $Q_{2}$ such that (1) $\overline{V}_{Q_{2}}$ either contains vertices of type only $v_{2}$ or (2) if it contains vertices of type $v_{1}$ and $v_{2}$ then it contains $v_{2}$ vertices with more multiplicity then vertices of type $v_{1}$.

We choose the following quadrangulation to represent $q_{2}$.
\begin{flalign}
Q_{2}\, =\, \{\, (1,4), (4,7), (7,10)\, \}
\end{flalign}
The vertex poset of ${\cal AC}_{Q_{1}}, {\cal AC}_{Q_{2}}$  is given in the appendix of \cite{Banerjee:2018tun}. The poset structure is crucial to compute the multiplicity of any vertex in $\overline{V}_{Q_{i}}$ or equivalently any configuration of propagators $\{\, s_{i_{1} \dots i_{k_{1}}},\, s_{j_{1} \dots j_{k_{2}}},\, s_{l_{1} \dots l_{k_{3}}}\, \}$. 
\subsection{Computing Multiplicities}
We now compute the multiplicity of a vertex whose quadrangulation corresponds to labeled graph $\Gamma_{q_{1}}$ in the set  $\overline{V}_{Q_{1}},\, \overline{V}_{Q_{2}}$ respectively. 

(a) Consider first any configuration corresponding to a quadrangulation $Q_{1}^{\prime}\, \in\, q_{1}$. We claim that given any one of the 12 vertices, say $v_{1}$ in $V({\cal AC}_{Q_{1}})$  whose labeled graph is $\Gamma_{q_{1}}$,  there always exists at least one permutation $\sigma$ such that in the poset associated to $V({\cal AC}^{\sigma}_{Q_{1}})$, $v_{1}$ is mapped to the configuration corresponding to quadrangulation $Q_{1}^{\prime}$.  This follows simply from the fundamental property of the equivalence class $\Gamma_{q_{1}}$, all of whose elements can be mapped onto each other by at least one $\sigma\, \in\, S_{10}$. 

(b) Let the vertex corresponding to a quadrangulation $Q_{1}^{\prime}$ in ${\cal K}_{n}$ be the following.
\begin{flalign}
v_{Q_{1}^{\prime}}\, =\, \{\, s_{i_{1}\, \dots\, i_{3}},\, s_{j_{1} \dots j_{5}},\, s_{m_{1} \dots m_{3}}\, \}
\end{flalign}
There are $3!^{2} \times 4!^{2}$  permutations that keep $v_{Q_{1}^{\prime}}$ fixed.  

Using (a), (b), we see that in $\overline{V}_{Q_{1}}$, the multiplicity of any vertex whose quadrangulation $Q_{1}^{\prime}\, \in\, \Gamma_{q_{1}}$ is $(\, (3 !)^{2}\, \times\, 4^{2}\, )\, 12$. In the same spirit, we can deduce that, 
\begin{flalign}
\begin{array}{lll}
\textrm{Mult. of}\, Q_{2}^{\prime}\, \in\, \Gamma_{q_{2}}\, \textrm{in}\, \overline{V}_{Q_{1}}\, =\, 3!^{4}\, \times\, 2\\
\textrm{Mult. of}\, Q_{1}^{\prime}\, \in\, \Gamma_{q_{1}}\, \textrm{in}\, \overline{V}_{Q_{2}}\, =\, 3!^{2}\, \times\,  4^{2}\, \times\, 8\\
\textrm{Mult. of}\, Q_{2}^{\prime}\, \in\, \Gamma_{q_{2}}\, \textrm{in}\, \overline{V}_{Q_{2}}\, =\, 3!^{4}\, \times\, 4
\end{array}
\end{flalign}
Based on these computations of various multiplicities, consider the following weighted sum over projective 3-forms evaluated on the deformed associahedra. 

Let $T_{i}$ be any arbitrary triangulations that contain $Q_{i}$ for $i\, \in\, \{\, 1,2\, \}$.  Let ${\cal M}_{10}(\sigma, i)\vert_{i=1}^{2}$ be a rational function that is indexed by the set of labeled graphs and which is defined via the following formula. 
\begin{flalign}
\Omega^{\sigma}_{Q_{i}}\vert_{A_{7}^{T_{i}, \sigma}}\, =\, {\cal M}(\sigma, i)\, \wedge_{(i,j)\, \in\, Q_{i}}\, d s_{\sigma(i)\, \dots\, \sigma(j)}
\end{flalign}
Then by computing the multiplicity of any vertex which is of the type $q_{1}$ or $q_{2}$ we get, 
\begin{flalign}
\sum\, \alpha_{i}\, \sum_{\sigma\, \in\, S_{10}}\, {\cal M}(\sigma, i)\, =\, {\cal M}_{10}(p_{1},\, \dots,\, p_{10})
\end{flalign}
where 
\begin{flalign}
\begin{array}{lll}
\alpha_{1}\, =\, \frac{1}{3!^{2}\, \cdots\, 16\, \cdot\, 12 + 3!^{4}\, \cdot\, 2}\\
\alpha_{2}\, =\, \frac{1}{3!^{2}\, \cdots\, 16\, \cdot\, 8 + 3!^{4}\, \cdot\, 4}
\end{array}
\end{flalign}
As a curiosity, we note that $\frac{\alpha_{1}}{\alpha_{2}}\, =\, \frac{34}{33}$ which is rather close to $1$.  Note that the issue of unequal multiplicity is resolved in this case by considering a weighted sum of the forms over both the labeled graphs.

We now conjecture a formula for the tree-level $n$ point amplitude in $\phi^{4}$ theory.  We note that to the combinatorics complexity involved in computing the multiplicity of each vertex in $\overline{V}$ (defined in equation (\ref{overlineV})), proof of this formula is beyond the scope of this paper. 

Let  $\Gamma_{q_{1}}, \dots,\, \Gamma_{q_{k(n)}}$ be the set of \emph{all} labelled graphs with representatives $Q_{1},\, \dots,\, Q_{k(n)}$. 

Based on empirical observations made in $n\, =\, 6,\, 8,\, 10$ point case, we conjecture the following formula for the generic  $n$ point amplitude.  
\begin{tcolorbox}[colback=black!5!white,colframe=black!75!black]
There exists a set of rational numbers $\{\, \alpha_{1},\, \dots,\, \alpha_{k(n)}\, \}$ such that, 
\begin{flalign}\label{cpfan}
{\cal M}_{n}(p_{1},\, \dots,\, p_{n})\, =\, \sum_{i=1}^{k(n)}\, \alpha_{i}\, {\cal M}(\sigma, i)
\end{flalign}
\end{tcolorbox}

\section{Cluster Polytopes and Non-planar One loop Integrands: An obstruction}\label{1loop}
 After the seminal work by Salvatori, \cite{Salvatori:2019aa}  where the search for a positive geometry associated with one loop S-matrix of bi-adjoint $\phi^{3}$ theory was first initiated,  Arkani-Hamed, He, Salvatori and Thomas (AHST) discovered a convex realization of the $D_{n}$ cluster-polytope, \cite{Arkani-Hamed:2019vag}.  AHST realization was in  the positive quadrant of a space spanned by Mandelstam invariants generated by external momenta and the loop momentum $l^{\mu}$. We denote this space as ${\cal K}_{n}^{1-l}$. It contains, as a proper subspace, the vector space spanned by $\{\, X_{ij}, p_{i}\, \cdot\, l,\, l^{2}\, \}$, \cite{Arkani-Hamed:2019vag, Jagadale:2022rbl}. Perhaps the simplest way to understand the one loop kinematic space for $n$ particles is to start with a set of $2n$ external momenta $\{\, p_{1},\, \dots,\, p_{n},\, p_{\overline{1}},\, \dots,\, p_{\overline{n}}\, \}$ that satisfy, 
 \begin{flalign}
 \sum_{i=1}^{n}\, p_{i}\, +\, \sum_{i=1}^{n}\, p_{\overline{i}}\, =\, 0
 \end{flalign}
 The physical kinematic space can be understood as a subspace in which $p_{\overline{i}}$ is identified with $p_{i}$. The kinematic space in which cluster polytopes live is spanned by planar kinematic variables of the type, 
\begin{flalign}
{\cal K}_{n}^{1-L}\, =\, \textrm{span}\{\, X_{ij}, X_{i\overline{j}},\, X_{\overline{i}\overline{j}}, X_{i\overline{i}}\, Y_{i}, Y_{\overline{i}}\, \}
\end{flalign}
where
\begin{flalign}
\begin{array}{lll}
Y_{i}\, =\, (p_{1}\, +\, \dots\, +\, p_{i-1}\, +\, l)^{2}\\
Y_{\overline{i}}\, =\, (p_{1}\, +\, \dots\, +\, p_{n}\, +\, p_{\overline{1}}\, +\, \dots\, +\, p_{\overline{i-1}}\, +\, l)^{2}
\end{array}
\end{flalign}
The ``doubling of external momenta" fits in beautifully with the pseudo-triangulation model for $D_{n}$ cluster polytope, which was proposed by Ceballos and Pilaud in \cite{Ceballos2015ClusterAO}. In this model, every planar kinematic variable is a chord that can dissect a $2n$-gon with an annulus in the middle. The chords associated with $Y_{i}$ are non-linear and terminate at a boundary point $0_{R}$ of the annulus, whereas the chords associated to $Y_{\overline{i}}$ terminate at the antipodal point $0_{L}$ of the annulus.  As shown in the figure (\ref{example}), the chords $X_{i\overline{i}}$ are non-linear and enclose the annulus such that any (pseudo)-triangulation of a 2n-gon which includes a chord $X_{i\overline{i}}$ must include $Y_{i}, Y_{\overline{i}}$. 
\begin{figure}
     \centering
     \includegraphics[scale=0.20]{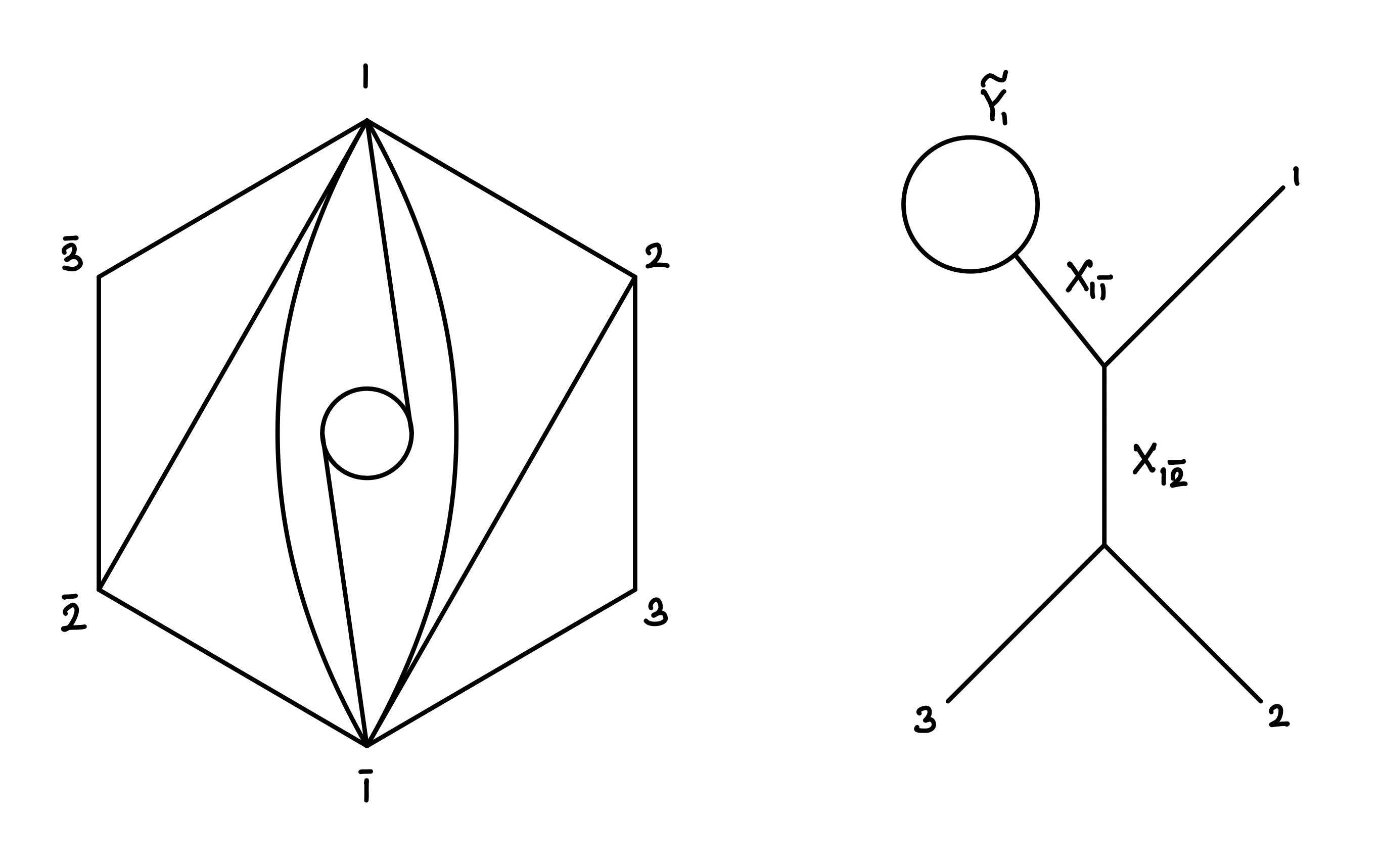}
     \caption{A closed polytope with all three channels}
     \label{example}
 \end{figure}
We have given a very brief summary of the construction of ${\cal K}_{n}^{1-l}$ and urge the reader to consult \cite{Jagadale:2022rbl} for a more elaborate discussion. 

The state of the art in the subject is the discovery of a new class of combinatorial polytopes called $\hat{D}_{n}$ polytopes  discovered by Arkani-Hamed,  Frost, Plamondon, Salvatori, and Thomas in \cite{afpst}. Boundaries of the convex realization of $\hat{D}_{n}$ correspond to \emph{all} the pseudo-triangulations of the $2n$-gon.\footnote{A brief review of $\hat{D}_{n}$ polytope and its AHST inspired realisation can be found in \cite{Jagadale:2022rbl}. We are indebted to Nima Arkani-Hamed for patiently explaining their construction to us.} These polytopes sit inside positive quadrant of ${\cal K}_{n}^{1-l}$.  

Although a detailed derivation of the convex realization is not essential for us (and can be found in\cite{Jagadale:2022rbl}), the basic idea is rather simple. In a nut-shell the convex realisation of $\hat{D}_{n}$ polytopes mirror the convex realisation of the associahedron $A_{n-3}^{T}$ in ${\cal K}_{n}$.

Given a reference (pseudo)-triangulation $PT_{0}$ of the holed $2n$-gon, we consider all the chords which belong to $PT_{0}^{c}$ which is a pseudo-triangulation obtained by a counter-clockwise $\frac{\pi}{n}$ rotation of $PT_{0}$ and write all the linear equations of the form
\begin{flalign}
Y_{IJ}\, =\, -\, c_{IJ}\, \hspace{0.1cm} \forall\,  \hspace{0.1cm} (I,J)\, \notin\, PT_{0}^{c}
\end{flalign}
where $(I,J)$ indicates all possible linear as well as non-linear chords and 
\begin{displaymath}
X_{IJ}\, =\, \{X_{ij}, X_{i\overline{j}}, X_{i\overline{i}}, Y_{i}, \tilde{Y}_{i}\, \}
\end{displaymath}
are the corresponding co-ordinates in the  ${\cal K}_{n}^{1-l}$. The projective $d\ln$ form on ${\cal K}_{n}^{1-l}$ uniquely determined by $\hat{D}_{n}$ generates one loop integrand of bi-adjoint $\phi^{3}$ amplitude, 
\begin{flalign}
m_{n}^{1-l}(p_{1}, \dots, p_{n}, l)\, \bigwedge_{(I,J)\, \in\, PT_{0}}\, \dd X_{IJ}\, =\, \Omega_{n}^{PT_{0}}\vert_{\hat{D}_{n}^{PT_{0}}}
\end{flalign}
It is tempting to consider the deformed realizations of the $\hat{D}_{n}$ polytopes where the deformation is parametrized by $[\sigma]\, \in\, {\cal G}_{n}$ (or more broadly, $\sigma\, \in\, S_{n}$) and see if the sum over all corresponding canonical forms is S-matrix integrand (at one loop) for $\phi^{3}$ theory without color. However, a moment of meditation informs us that a naive application of such an idea can not work. This can be understood by a simple counting argument for $n\, =\, 4$.  

The cardinality of the vertex set of $\hat{D}_{4}$ is 70, and this can be seen as follows. As each vertex of $\hat{D}_{n}$ corresponds to a unique pseudo-triangulation of the holed $2n$-gon, we need to simply count the number of pseudo-triangulations. However, just as for triangulation, each pseudo-triangulation is dual to a planar  one-loop Feynman graph with cubic vertices, such that the loop momentum is oriented clockwise or counter-clockwise. For a fixed orientation of the loop momentum, the number of such graphs is twenty graphs with tadpoles, ten graphs with a bubble, four graphs associated with vertex corrections, and one box graph. Hence,
\begin{flalign}
\vert \textrm{1-loop planar oriented Feynman graph}\, \vert\ =\, \vert\, V(\hat{D}_{n})\, \vert\, =\, 70
\end{flalign}
From the perspective of amplitudes, the map between a set of all vertices of $\hat{D}_{n}$ polytope with the set of planar 1-loop Feynman graphs is $2: 1$. In other words two vertices $v, v^{\prime}\, \in\, V(\hat{D}_{n})$ are equivalent if 
\begin{flalign}
\forall\, (i,j)\, \in\, v\, \exists\, (\overline{i}, \overline{j})\, \in\, v^{\prime}\, \textrm{or}\, \forall\, (i,0_{R})\, \in\, v\, \exists\, (i,0_{L})\, \in\, v^{\prime}
\end{flalign}
Modulo this equivalence, we  see that the number of ``independent" vertices in $\frac{\hat{D}_{n}}{Z_{2}}$ is 35. 
\begin{flalign}
\vert\, \cup_{[\sigma]\, \in\, {\cal G}_{4}}\, V(\hat{D}_{n}^{\sigma})\, \vert\, =\, 105
\end{flalign}
On the other hand, the total number of one loop 4-point Feynman graphs with cubic vertices is 54.\footnote{Once again, this can be verified by a simple counting argument.} We can contrast this situation with associahedron where
\begin{flalign}
\vert\, \cup_{[\sigma]\, \in\, {\cal G}_{n}}\, V( A_{n-3}^{T, \sigma})\, \vert\, =\, 2^{n-3}\, \vert\, \textrm{Feynman graphs}\, \vert
\end{flalign}
Hence we see that the cardinality of the complete vertex set spanned by ${\cal G}_{4}$ action on AHST realization of $\hat{D}_{4}$ is not an integer multiple of the total number of  Feynman graphs $\phi^{3}$ theory. We thus conclude that if we define the $[\sigma]$ dependent rational function on ${\cal K}_{4}^{1-l}$ via,
\begin{flalign}
\Omega_{4}^{PT_{0}, \sigma}\vert_{\hat{D}_{4}^{PT_{0},\sigma}}\, \rightarrow\, {\cal M}_{4}^{1-l}(\sigma)
\end{flalign}
then there exists no $\alpha_{4}\, \in\, {\bf Z}^{+}$ for which 
\begin{flalign}
\sum_{[\sigma]\, \in\, {\cal G}_{4}}\, {\cal M}_{4}(\sigma)\, =\, \alpha_{4}\, m_{4}^{1-l}(p_{1}, \dots,\, p_{4}).
\end{flalign}
It can be easily argued that this result continues to hold $\forall\, n$. 

Thus there is an obstruction in realizing $m_{n}^{1-l}$ in scalar theory without color as a (sum of) canonical forms of $\hat{D}_{n}$ polytopes. In fact, in the $n = 4$ case, we could have foreseen this already. While every vertex of $\hat{D}_{4}$ which is not dual to box diagram occurs twice in the set of all vertices, the graphs with no tree-level pole (that is, box graphs for three orderings $\{\, (1,2,3,4),\, (1,3,2,4),\, (1,2,4,3)\, \}$) occur only once. In fact, this should be expected as the set of vertices of $\hat{D}_{4}$ not only correspond to the color ordering of external states but also planar loops.  However, in the perturbative expansion of the S-matrix, once the external states are not colored, then even for a fixed ordering, we can have the box as well as cross-box graphs. Such non-planar diagrams do not correspond to any vertex of the $\hat{D}_{n}$. 
Several comments are in order.
\begin{itemize}
\item It can be verified that there is no weighted sum over canonical forms associated with $\hat{D}_{4}$ and the deformed realization of the 4-dimensional cyclohedron $\hat{C}_{4}$ which is proportional to the one loop integrand. This is because every vertex of a cyclohedron is dual to a tadpole graph with a cubic vertex. Hence although the ratio of residues between the tadpole graphs and the box graphs can be adjusted by suitably changing relative weights of $\hat{D}_{4}$, $\hat{C}_{4}$ forms, this does not change the residue of a vertex associated to one loop propagator.  
\item We believe that more general Clusterohedra discovered in \cite{afpst}, whose boundary poset includes planar as well as non-planar poles of the one loop integrand in bi-adjoint theory will play a crucial role in the hunt for positive geometry of amplitudes without color. 
\end{itemize}

\section{Conclusions and Open questions}
In this note, we have continued to develop the ideas proposed in \cite{Jagadale:2021iab, Jagadale:2022rbl}. The central premise behind these ideas really should be thought of morally as a ``bootstrap" construction in the context of the positive geometry program of the S-matrix. Namely, under what conditions do  convex realizations of positive geometries (whose boundary poset is isomorphic to a set of poles of an S-matrix) geometrize the S-matrix of a local unitary QFT? Although for an arbitrary diffeomorphism, the answer is not known, we have shown that for an infinite family of linear diffeomorphisms, the resulting realizations through their canonical forms always define a tree-level S matrix.\footnote{In fact as we proved in \cite{Jagadale:2022rbl}, there is a class of diffeomorphisms which deform the $\hat{D}$ polytope in such a way that the resulting realisation constitute a positive geometry for one loop integrand of a scalar field theory in which two scalars with unequal masses interact via a cubic coupling.}  In this paper we analysed a rather  canonical choice of diffeomorphisms that arise from the combinatorial Bose symmetry acting on the configuration space of $n$ momenta. The resulting associahedra collective constitutes a positive geometry for tree-level S matrix without color. 

In a beautiful paper \cite{Damgaard:2021ab}, the authors have shown how  the Kleiss-Kuijf relations emerge from the geometry of the positive geometries, such as the momentum amplituhedron and the associahedron. One of the corollaries of their analysis is to obtain any channel (or a collection of channels) of the $\phi^{3}$ amplitude as the boundary of an open associahedron in ${\cal K}_{n}$. All of the associahedra lie in the fixed linearity space spanned by $\{\, X_{i_{1}j_{1}},\, \dots,\, X_{i_{n-3},j_{n-3}}\, \}$ in ${\cal K}_{n}$. The S-matrix of uncolored $\phi^{3}$ theory can then be obtained as an oriented sum over positive geometries (for a precise definition of the oriented sum over positive geometries, see \cite{Damgaard:2021ab}, \cite{Dian:2022tpf}). However, the positive geometries that will generate the S-matrix are distinct from the deformed realizations we have found in our paper.  However it will be interesting to compare the two approaches and see if the idea proposed in their work can be applied to obtain the S-matrix of $\phi^{p}$ theory from positive geometries. 

 The central result of this paper is a rather simple formula for the S-matrix, which was proved using a remarkable combinatorial formula relating the Catalan number, the number of tri-valent graphs with $n$ external vertices, and the cardinality of the symmetric group. We then argued that restriction of $\Omega_{n}^{Q}$ on the ABHY associahedron generates the S-matrix of $\phi^{4}$ theory without color. Although a complete formula expressing $n$ point amplitude of $\phi^{4}$ theory (as a weighted sum of lower forms pulled back onto the deformed associahedra) is beyond the scope of this paper, we believe that the conceptual setup basically will go through for all accordiohedra and hence any $\phi^{p}$ theory. 
 
The results in this paper, along with those in \cite{Jagadale:2021iab, Jagadale:2022rbl} are data points that reveal the striking universality of ABHY associahedron as a geometric structure associated with the S-matrix. While ABHY realization gives a very specific shape to associahedron in the embedding space, it corresponds to an infinite family of realizations all of which are diffeomorphic to the ABHY associahedron and a large family of these are related to increasingly complicated quantum field theories which are unitarily inequivalent. 
 
There are many potential avenues to develop these ideas further including the hunt for positive geometry for the one loop integrands of $\phi^{ 3}$ S-matrix without color. 

We finally close this note with a rather trivial observation, but which hints at the relevance of positive geometries arising via gluing of associahedra. 
\subsection*{Can all the Deformations be Combined into a Single Geometry?}
In the n =4 case, we can take the mirror images of the three associahedra and join them at the respective vertices to obtain a closed one-dimensional hexagon in which the adjacent edges have the opposite orientation.  (see figure \ref{figurex}). 

\begin{figure}
     \centering
     \includegraphics[scale=0.20]{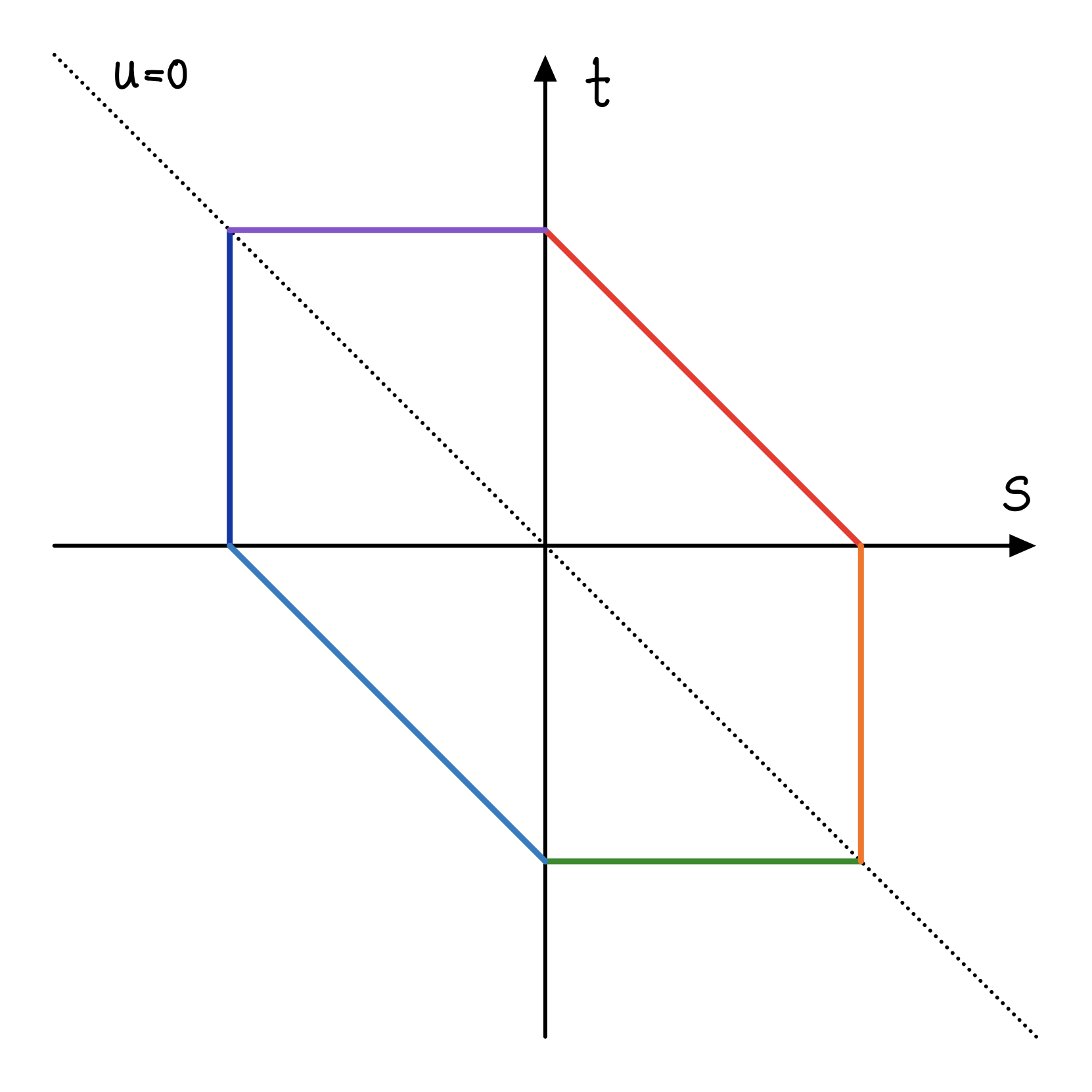}
     \caption{A closed polytope with all three channels}
     \label{figurex}
 \end{figure}
The resulting canonical form then has simple poles at all the six vertices, and the corresponding one form is $4\, (\frac{1}{s} + \frac{1}{t} + \frac{1}{u} )$. It will be interesting to investigate if such a structure persists at higher $n$.  
\section*{Acknolwedgements}
We would like to thank Pinaki Banerjee, Nikhil Kalyanapuram, Subramanya Hegde, Arkajyoti Manna, Prashanth Raman, Aninda Sinha and Ashoke Sen for discussions. We are especially thankful to Nima Arkani-Hamed for stimulating discussions in the initial phase of this project and his constant support over the years. AL would like to thank Center for High Energy Physics (CHEP) at the Indian Institute of Science, International Center for Theoretical Physics (ICTS) and Department of Theoretical Physics (DTP) at Tata Institute of Fundamental Research (TIFR), Mumbai  for their hospitality at various stages of this project. MJ is supported by the Walter Burke Institute for Theoretical Physics, the U.S. Department of Energy, the Office of Science, Office of High Energy Physics under Award No. DE-SC0011632.

\bibliographystyle{unsrt}
\bibliography{Bibliography}

\end{document}